\documentclass[aps,prx,reprint,superscriptaddress,twocolumns,longbibliography]{revtex4-2}
\usepackage{graphicx} 
\usepackage{amsmath}
\usepackage{amsfonts}
\usepackage{mathrsfs}
\usepackage{amsthm}
\usepackage{bm}
\usepackage{float}
\usepackage{xcolor}
\usepackage{hyperref}
\usepackage{soul}

\usepackage[ruled,vlined,linesnumbered]{algorithm2e}

\providecommand{\rmd}{\ensuremath{\mathrm{d}}}

\providecommand{\calL}{\ensuremath{\mathcal{L}}}
\providecommand{\calM}{\ensuremath{\mathcal{M}}}
\providecommand{\calN}{\ensuremath{\mathcal{N}}}
\providecommand{\calO}{\ensuremath{\mathcal{O}}}

\providecommand{\calT}{\ensuremath{\mathcal{T}}}
\providecommand{\calU}{\ensuremath{\mathcal{U}}}

\providecommand{\bbE}{\ensuremath{\mathbb{E}}}

\newcommand{\ket}[1]{\left\lvert #1 \right\rangle}
\newcommand{\bra}[1]{\left\langle #1 \right\rvert}
\newcommand{\braket}[2]{\left\langle #1 \middle| #2 \right\rangle}
\newcommand{\ketbra}[2]{\left\lvert #1 \right\rangle \! \left\langle #2 \right\rvert}
\newcommand{\ketbrat}[1]{\left\lvert #1 \right\rangle \! \left\langle #1 \right\rvert}
\newcommand{\norm}[1]{\left\lVert#1\right\rVert}
\newcommand{\abs}[1]{\left\lvert#1\right\rvert}
\newcommand{\set}[1]{\left\{ #1\right\}}

\DeclareMathOperator{\Tr}{Tr}
\newcommand{\trace}[2]{\Tr_{#1}\left( #2 \right)}
\DeclareMathOperator*{\E}{{\mathbb{E}}}
\newcommand{\Exp}[2]{\E_{#1}\left[ #2 \right]}

\newtheorem{theorem}{Theorem}
\newtheorem{definition}[theorem]{Definition}

\newtheorem{lemma}[theorem]{Lemma}

\newtheorem{corollary}[theorem]{Corollary}

\newcommand{\certified}{\textsc{Certified}}
\newcommand{\failed}{\textsc{Failed}}

\newcommand{\id}{\mathbb{I}}

\newcommand{\htheta}{\hat{\theta}}
\newcommand{\bbraket}[1]{\langle  #1 \rangle}
\newcommand{\opt}{{\rm opt}}

\newcommand{\tr}{{\mathrm{Tr}}}

\newcommand{\mH}{{\mathcal{H}}}

\newcommand{\ttrace}{{\rm Tr}}

\newcommand{\bR}{{\mathbb{R}}}

\newcommand{\bE}{{\mathbb{E}}}

\definecolor{gold}{rgb}{0.85, 0.65, 0.13}
\usepackage[dvipsnames]{xcolor}

\begin{document}

\title{Near-optimal quantum metrology with few-qubit measurements}

\author{Liang Mao}
\affiliation{Institute for Advanced Study, Tsinghua University, Beijing, China}

\author{Senrui Chen}
\affiliation{Institute for Quantum Information and Matter, California Institute of Technology, Pasadena, CA 91125, USA}

\author{Hsin-Yuan Huang}
\affiliation{Oratomic, Pasadena, California 91125, USA}
\affiliation{Institute for Quantum Information and Matter, California Institute of Technology, Pasadena, CA 91125, USA}

\author{John Preskill}
\affiliation{Institute for Quantum Information and Matter, California Institute of Technology, Pasadena, CA 91125, USA}
\affiliation{Oratomic, Pasadena, California 91125, USA}

\author{Sisi Zhou}
\affiliation{Perimeter Institute for Theoretical Physics, Waterloo, Ontario N2L 2Y5, Canada}
\affiliation{Department of Physics and Astronomy, Department of Applied Mathematics and Institute for Quantum Computing, University of Waterloo, Ontario N2L 2Y5, Canada}

\date{\today}

\begin{abstract}
    Quantum metrology, which addresses parameter estimation in quantum systems, has broad applications across 
    science and technology. Conventional metrology protocols for multi-qubit states in the multi-parameter regime typically require highly complex quantum measurements, leading to substantial quantum-resource costs. In this work, we introduce a family of metrology protocols that use only few-qubit measurements, thereby significantly reducing the required resources. For arbitrary pure states, one of our protocols 
    approaches the quantum Cram\'{e}r-Rao bound up to an overhead in sample complexity that scales linearly with the number of qubits, irrespective of the number of parameters to be estimated. For typical Haar-random states, this overhead can be reduced to a constant. 
   Our results build on recent advances in quantum state certification protocols with few-qubit measurements: we establish a universal connection between certification and metrology in which the precision of the certification protocol determines the metrological overhead.
    We also illustrate our approach through an example of Hamiltonian estimation from ground states. 
\end{abstract}
\maketitle


\emph{Introduction.---}
Quantum precision measurement is arguably one of the most 
consequential applications of quantum science and technology.
In situations ranging from optical interferometry~\cite{caves1981quantum,yurke19862,ligo2011gravitational,ligo2013enhanced} and frequency detection~\cite{wineland1992spin,bollinger1996optimal,leibfried2004toward,taylor2008high,rosenband2008frequency,appel2009mesoscopic,ludlow2015optical,zhou2020quantum,kaubruegger2021quantum,marciniak2022optimal} to quantum imaging~\cite{le2013optical,lemos2014quantum,tsang2016quantum,abobeih2019atomic}, the task can be reduced to estimating unknown parameters encoded in a quantum state. Quantum metrology provides a unified framework for this parameter-estimation problem~\cite{giovannetti2004quantum,giovannetti2006quantum,giovannetti2011advances,degen2017quantum,pezze2018quantum,pirandola2018advances}. In the setting where the parameters are approximately known and the goal is to refine them,
the quantum Cram\'{e}r-Rao bound (CRB)~\cite{rao1945information,cramer1999mathematical,lehmann2006theory,holevo2011probabilistic,helstrom1967minimum,helstrom1968minimum,helstrom1976quantum,braunstein1994statistical,barndorff2000fisher,paris2009quantum} characterizes the ultimate precision of any estimator through an information-theoretic quantity called the quantum Fisher information (QFI). In the single-parameter regime, the QFI is simply the classical Fisher information (CFI) maximized over all possible quantum
measurements on quantum states, and optimal measurements that attain the quantum CRB are well established~\cite{braunstein1994statistical}.

In the multi-parameter regime, matters become considerably more complicated. Since optimal measurements for different parameters do not necessarily coincide, the quantum CRB is generally not attainable. Substantial efforts have been made to attain the bound as closely as possible. It has been shown that by leveraging complex quantum operations, such as multi-copy measurements for mixed states~\cite{kahn2009local,yamagata2013quantum,yang2019attaining} and randomized measurements for pure states~\cite{
hayashi1998asymptotic,zhu2014quantum,zhu2018universally,hou2018deterministic,lu2025quantum,zhou2026randomized,du2026complexity}, one can approach the quantum CRB within a constant factor.
However, these operations may not be
 within the direct capabilities of specialized quantum sensing devices.
Even for programmable quantum platforms, these measurements are typically of \emph{high complexity}. Moreover, in practice errors grow rapidly with circuit depth, necessitating integration with fault-tolerant schemes~\cite{arrad2014increasing, kessler2014quantum, ozeri2013heisenberg, dur2014improved, demkowicz2017adaptive, zhou2018achieving, layden2018spatial, layden2019ancilla, chen2024quantum, mann2025quantum, antu2025stabilizer,liu2026subsystem,Pradenne2026restrictions}, whose feasibility remains an active research question.
These challenges motivate the central question of this work:
\begin{center}
    \emph{Can few-qubit measurements nearly attain the optimal metrological precision set by the quantum CRB?}
\end{center}

We answer this question in the affirmative by showing a universal connection between quantum state certification protocols and metrology protocols.
In quantum state certification
 the task is to verify whether the fidelity between an unknown state and a specified target state is sufficiently high~\cite{flammia2011direct,montanaro2013survey,pallister2018optimal,buadescu2019quantum,zhu2019optimal,kliesch2021theory,li2025universal}. Recent progress~\cite{huang2025certifying,gupta2025few,coladangelo2026power,du2025certifying,coladangelo2026robust} demonstrates that few-qubit measurements are capable of certifying quantum states when the target states are pure. The protocols used in these works share a common structure which we term \emph{conditional fidelity estimation} (CFE).
 We prove that all CFE-based certification protocols can be adapted into metrology protocols, 
 with performance characterized by the certification gap, which quantifies the resolution of quantum state certification. Thus, drawing on established results in certification, we provide a series of metrology protocols for parameterized pure states~\footnote{For generic pure states, unlike for pure states, it is impossible to attain the quantum CRB using only measurements on single copies of states, and the gap can scale polynomially with the system dimension. See Appendix.~\ref{app:review} for a detailed review.} that utilize only few-qubit (or even single-qubit) measurements. 
 
We provide three metrology protocols. The first protocol combines computational-basis measurement on most qubits with a random Pauli measurement on one randomly chosen qubit. Despite its extreme simplicity, we prove that this protocol achieves the quantum CRB up to a small polynomial factor in the number of qubits $n$ for many quantum states of interest: The factor is $\calO(n^2)$ for typical Haar-random states and $\calO(n)$ for more structured states such as gapped ground states. With access to classical adaptivity, the second protocol improves upon the first by using a more structured decision-tree (DT) product basis, and achieves a provable $\calO(n)$ factor over the quantum CRB for \emph{all} pure quantum states. 
Finally, when targeting only typical Haar-random states, we provide the third protocol that uses randomized Pauli measurement and attains the quantum CRB up to a \emph{constant} factor. 
We say these protocols are \emph{near-optimal}~\footnote{Unlike Ref.~\cite{zhou2026randomized} which defines the near-optimal metrology protocols to be protocols that attain the quantum CRB up to a constant factor, we slightly generalize the definition to include a polynomial factor (in the number of qubit), which is still exponentially small compared to the system dimension.}, meaning that they attain the quantum CRB up to polynomial factors, i.e. with polynomial overhead in sample complexity.


\begin{figure}[t!]
    \centering
    \includegraphics[width=0.9\linewidth]{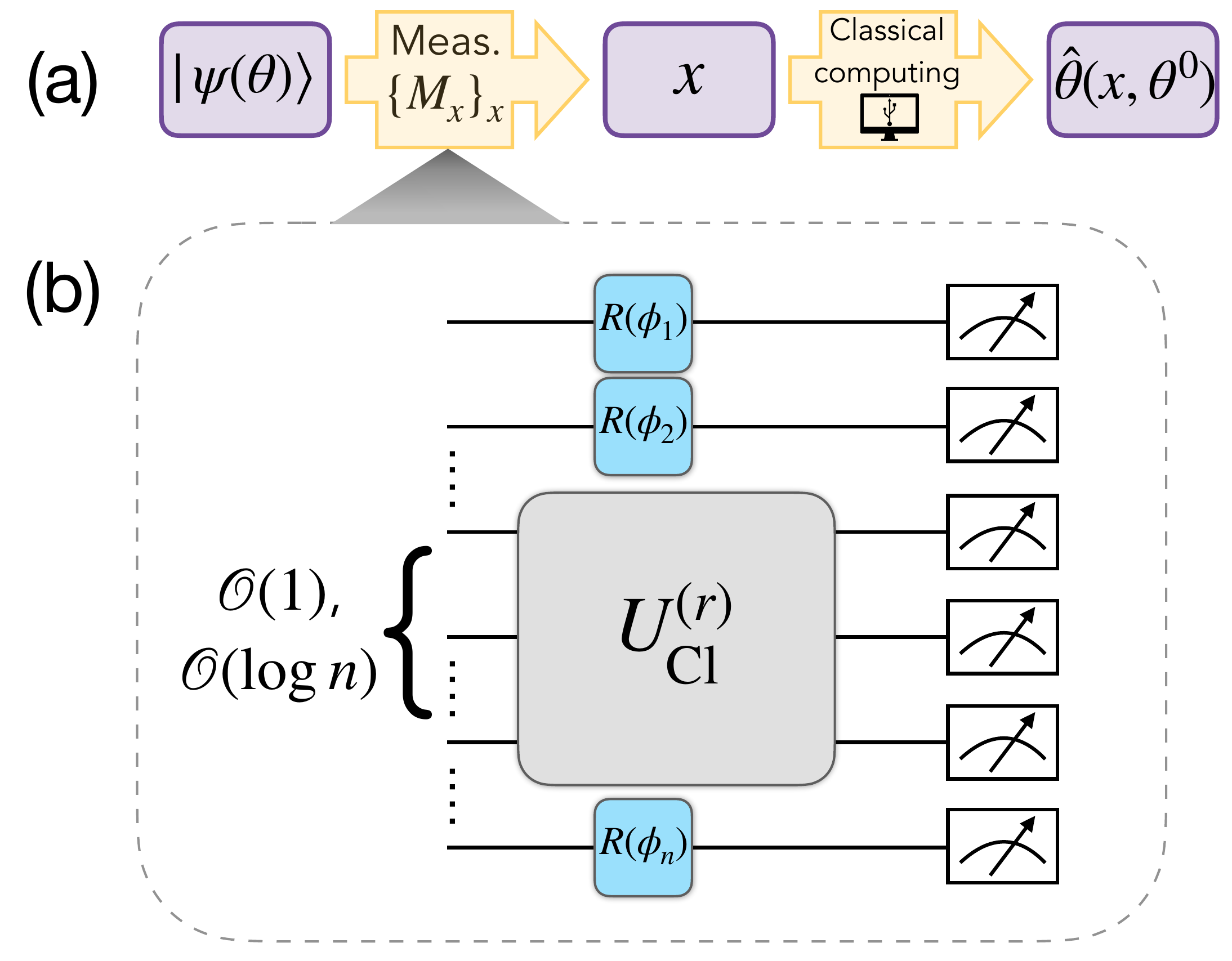}
    \caption{Schematic illustration of quantum metrology and our measurement protocol. (a) 
    In quantum metrology, a POVM $\set{M_x}_x$ is performed on the state $\ket{\psi(\theta)}$. Upon the outcome $x$, classical computer builds an estimator $\hat{\theta}$. In quantum metrology, one typically has access to a prior estimate $\theta^0$ that is close to the true value $\theta$. (b) In our protocol, the POVM acts locally on $n-r$ qubits and as a random Clifford measurement on $r$ qubits. Throughout the protocols introduced this work, $r$ is at most $\calO(\log n)$. In most cases, $r=1$ is sufficient.}
    \label{fig:illustration}
\end{figure}

We also illustrate our first protocol by applying it to estimating a Hamiltonian when given access to multiple copies of its ground state. By numerically simulating a leading-order maximum likelihood estimator (MLE)~\cite{fisher1922mathematical,van2000asymptotic} 
for a disordered transverse-field Ising model, we observe mean squared errors close to those given by the quantum CRB, confirming our theoretical prediction.

\emph{Quantum metrology.---}
Here we briefly review quantum metrology to set the context and notation for subsequent sections. Consider a $d$-dimensional parameterized quantum state $\rho(\theta)$ in Hilbert space $\mH$, where parameters are denoted by $\theta = (\theta_1,\theta_2,\ldots,\theta_m)$. Here $m$ is the number of parameters and $\Theta \subseteq \bR^m$ is the domain of $\theta$. An estimator $\htheta(x)$, given a corresponding positive operator-valued measurement (POVM) $\calM = \{M_x\}_{x}$, is a function that maps measurement outcomes $x$ to $\Theta$.  The performance of an estimator at $\theta$ is characterized by the \emph{mean squared error matrix} (MSEM, with MSE denoting the mean squared error)  
\begin{equation}
\label{eq:msem}
    V(\calM,\htheta)_{ij} = \sum_x (\htheta(x)_i - \theta_i)(\htheta(x)_j - \theta_j) \ttrace(\rho(\theta) M_x). 
\end{equation}
We focus on analyzing the performance of \emph{locally unbiased estimators} that are unbiased at the true value of $\theta$ and in its vicinity to first order. This is relevant in metrological situations where prior knowledge of $\theta$ is available in the vicinity of its true value~\cite{rao1973linear,kay1993fundamentals,lehmann2006theory,cox2017inference,casella2002statistical,gill2000state,yang2019attaining}. 
In classical statistics, the MSEM of any locally unbiased estimator at $\theta$ is bounded below by the inverse of the classical Fisher information (CFI) matrix, defined as
$
    I(\calM)_{ij} = \sum_{x, p_x(\theta):=\ttrace(\rho(\theta) M_x) \neq 0} \frac{1}{p_x(\theta)} \frac{\partial p_x(\theta)}{\partial \theta_i} \frac{\partial p_x(\theta)}{\partial \theta_j}. 
$
$I(\calM)$ is a function of the POVM and is independent of the estimator. We  assume throughout that $I(\calM)$ is strictly positive definite, meaning each parameter is identifiable.
The CRB~\cite{rao1973linear,kay1993fundamentals,lehmann2006theory,cox2017inference,casella2002statistical} states that for any locally unbiased estimator built from $\calM$,
$
    V(\calM,\htheta) \succeq I(\calM)^{-1},
$
where $A\succeq B$ means $A-B$ is positive semi-definite.
Here we assume $\htheta(x)$ uses only a single measurement result $x$. One can also perform $N$ i.i.d.\ measurements $\calM$ on multiple copies of the state and use all results to construct an estimator. In this case,
$
    V(\calM,\htheta(x_1,\ldots, x_N)) \succeq I(\calM)^{-1}/N.
$
The CRB is asymptotically attainable using MLE when $N \rightarrow \infty$.

The CRB above is a purely classical result for a fixed choice of POVM. The quantum CRB~\cite{helstrom1967minimum,helstrom1968minimum,helstrom1976quantum,holevo2011probabilistic,braunstein1994statistical,barndorff2000fisher,paris2009quantum} provides a more general lower bound when the POVM can be chosen freely. For any locally unbiased estimator of the target state $\rho$,
\begin{equation}
\label{eq:QCRB}
    V(\calM,\htheta,\rho(\theta)) \succeq I(\calM,\rho(\theta))^{-1} \succeq J^{-1}(\rho(\theta)),
\end{equation}
where $J$ is the quantum Fisher information (QFI) matrix as a function of $\rho(\theta)$. In this work, we focus on the case where $\rho$ is a pure state, for which the QFI matrix takes the explicit form
\begin{align}
    J_{ij} = 4\operatorname{Re}\left[
    \braket{\partial_i\psi}{\partial_j\psi}-\braket{\psi}{\partial_i\psi}\braket{\partial_j\psi}{\psi}
    \right].
\end{align}
There is a useful geometric interpretation of $J$ for pure states: $J$ is the metric of fidelity in parameter space,
\begin{multline}
    \abs{\bbraket{\psi(\theta)|\psi(\theta+\mathrm{d}\theta)}}^2
   =\\
   1-\frac{1}{4}\sum_{ij}J_{ij}(\ket{\psi(\theta)})\,\mathrm{d}\theta_i
    \mathrm{d\theta}_j + o(\norm{\rmd\theta}^2).\label{eq:geo-mean-qfi}
\end{multline}

For single-parameter estimation, i.e., when $J$ and $I(\calM)$ are scalars, $J$ is the maximum of $I$ over all POVMs~\cite{braunstein1994statistical}. Furthermore, single-qubit measurements with classical communication are known to be able to achieve $J$ for single-parameter pure states~\cite{zhou2020saturating}.  For multi-parameter estimation, the quantum CRB is not always attainable because measurements optimal for different parameters may not be compatible, i.e., cannot be implemented simultaneously. 
While impossible for generic mixed states, for pure states $I(\calM)$ can be close to $J$ up to constant overhead~\cite{
hayashi1998asymptotic,zhu2014quantum,lu2025quantum,zhou2026randomized,du2026complexity,li2016fisher}. However, no simple measurement protocols with \emph{few-qubit measurements} (see Fig.~\ref{fig:illustration}(b)) are known so far. See Appendix~\ref{app:review} for a more detailed review.

\emph{Relation to quantum state certification.---}
To develop few-qubit measurement protocols for metrology, we draw insights from quantum state certification.
The goal of state certification is to determine whether a state $\rho$ prepared in the lab is sufficiently close to the target state $\ket{\psi}$. Although projection onto $\ketbrat{\psi}$ can resolve this task, the implementation cost is typically high due to the complexity of $\ket{\psi}$. 
Recent results resolve this issue by providing alternative subroutines, demonstrating that few-qubit measurements suffice to certify quantum states~\cite{huang2025certifying,gupta2025few,coladangelo2026power,du2025certifying,coladangelo2026robust}.

In these protocols, we observe a common structure, which we term \emph{conditional fidelity estimation} (CFE). A CFE protocol first measures (possibly randomly selected) $n-r$ qubits of the lab state $\rho$ in one or more single-qubit bases 
with outcome $\ket{x}$, yielding a post-selected $r$-qubit lab state $\rho_x=\bra{x}\rho\ket{x}/\norm{\bra{x}\rho\ket{x}}$ corresponding to a post-selected target state $\ket{\psi_x}=\bra{x}\psi\rangle/\norm{\bra{x}\psi\rangle}$.
 Assuming the ability to query a description of $\ket{\psi_x}$ (e.g., query access to amplitudes of $\ket{\psi}$ in a product basis), 
 the protocol then estimates the conditional fidelity $\bra{\psi_x}\rho_x\ket{\psi_x}$ -- for example via classical shadows using random Clifford measurements, though any estimator of the conditional fidelity may be invoked. For our purposes, the crucial feature of CFE is that this averaged conditional fidelity provides a useful bound on the true fidelity:
 

 \begin{align}\label{main:eq:inf}
    \Exp{x}{1-\bra{\psi_x}\rho_x\ket{\psi_x}}\geq \frac{1-\bra{\psi}\rho\ket{\psi}}{\Delta},
\end{align}
where $\Delta>1$ is the \emph{certification gap}. For a predetermined precision $\epsilon$, this method can determine with high confidence whether the fidelity is greater than $1-\epsilon$ or less than $1-\Delta\epsilon$. See Appendix~\ref{app:certification} for details.

 Any certification protocol based on CFE can be transformed into a metrology protocol using the same POVM. In this work, we will usually assume the conditional fidelity is estimated using random Clifford measurements as illustrated in Fig.~\ref{fig:illustration}, because this approach has been well characterized in previous work~\cite{zhou2026randomized}.
Then the total implementation cost per experiment consists of single-qubit rotations on $n-r$ qubits and a random Clifford rotation on $r$ qubits, which requires $\calO(r^2)$ two-qubit Clifford gates to synthesize. Throughout this work, $r$ is 
at most $\calO(\log n)$, so the total implementation cost is 
low.

These measurement protocols have provable guarantees on the CFI. 
{For a pure lab state $\ket{\psi(\theta)}$ with post-selected state denoted by $\ket{\psi_{x}(\theta)}$}, expanding Eq.~\eqref{main:eq:inf} to quadratic order in $\theta$ and using Eq.~\eqref{eq:geo-mean-qfi}, we have
\begin{align}\label{eq:qfi-gap}
    \Exp{x}{J(\ket{\psi_x(\theta)})}\succeq \frac{J(\ket{\psi(\theta)})}{\Delta}.
\end{align}
(Note that the dependence on $\theta$ is sometimes suppressed in the following equations.)
For the few-qubit post-selected state $\ket{\psi_x}$, the Clifford measurement $\calU^{(r)}_{\mathrm{Cl}}$ is known to satisfy  $I^{-1}\big(\calU^{(r)}_{\mathrm{Cl}}\big)\preceq 4J^{-1}(\ket{\psi_x})$ for pure state metrology~\cite{zhou2026randomized}. Combining these results, we have
\begin{align}\label{eq:intuition}
    I(\calM,\ket{\psi})&=I\left(\set{p_x}_x\right)+\Exp{x}{I\big(\calU^{(r)}_{\mathrm{Cl}},\ket{\psi_x}\big)}\notag\\
    &\succeq \Exp{x}{I\big(\calU^{(r)}_{\mathrm{Cl}},\ket{\psi_x}\big)}\succeq  \frac{J(\ket{\psi})}{4\Delta},
\end{align}
where $p_x$ is the probability of measurement outcome $x$ on $n-r$ qubits~\cite{combes2014quantum}.
It demonstrates that we can transform CFE-based certification protocols into metrology protocols, with a provable performance guarantee.
\begin{lemma}[Informal]
    Given any CFE-based certification protocol that first measures $n-r$ qubits with measurement outcomes $x$ satisfying Eq.~\eqref{main:eq:inf}, 
    followed by random Clifford measurement
    on the remaining $r$ qubits, the resulting POVM $\calM$ satisfies
    \begin{align}\label{eq:lemma_4}
        I^{-1}(\calM,\ket{\psi})\preceq  4\Delta J^{-1}(\ket{\psi}),
    \end{align}
    where $\Delta$ is the certification gap.
\end{lemma}\noindent
In Appendix~\ref{app:certification}, we prove the formal version, Lemma~\ref{lem:main}, by explicitly constructing a set of locally unbiased estimators whose MSEM is upper bounded by Eq.~\eqref{eq:lemma_4}, establishing an even stronger result. We note that to construct a locally unbiased estimator at $\theta^0$, we need query accesses to both $\ket{\psi_x(\theta^0)}$ and $\partial_{\theta_i}\!\ket{\psi_x(\theta)}\!|_{\theta = \theta^0}$.   

\emph{Protocols.---}
Applying this lemma to established protocols for state certification~\cite{huang2025certifying,gupta2025few,du2025certifying,coladangelo2026robust}, we can construct three protocols for metrology that use only few-qubit measurements.
The first is a non-adaptive protocol with single-qubit computational basis measurements, following from the state certification protocol in Ref.~\cite{huang2025certifying}. Here the
 first step of the CFE is simply computational-basis measurement on $n-r$ qubits.  
 We primarily focus on the $r=1$ case, where the random Clifford measurement on the remaining qubit is reduced to a random Pauli measurement.
\begin{theorem}[Non-adaptive protocol; informal]
\label{main:thm:1}
    Consider a POVM where we uniformly randomly choose a qubit $i$ from $n$ qubits, measure it in a random Pauli basis $\alpha\in\set{x,y,z}$ (each with probability $1/3$), and measure the remaining $n-1$ qubits in the computational basis.
    The CFI matrix satisfies
    \begin{equation}\label{main:eq:bound}
        I^{-1}(\calM,\ket{\psi(\theta)})\preceq 4\tau(\theta)
        J^{-1}(\ket{\psi(\theta)}),
    \end{equation}
    where $\tau(\theta)$ depends on target state $\ket{\psi(\theta)}$.
\end{theorem} 
\noindent
Here $\tau(\theta) = \Delta$ is the mixing time of
a Markov chain on the Boolean hypercube defined by the computational-basis amplitudes of $\ket{\psi(\theta)}$ (see Theorem~\ref{thm:metrology1} in Appendix~\ref{app:protocol1} for details). 

Because the variance scales inversely with the number of copies of the lab state, Eq.~\eqref{main:eq:bound} means that our protocol reaches a target precision using $4\tau(\theta)$ times as many copies as a protocol saturating the quantum CRB.

Note that
there are quantum states for which the Markov chain defined by this $r=1$ protocol mixes very slowly or not at all. In such cases, we may modify the protocol so that the random Clifford measurement acts on $r>1$ qubits. Then the Markov chain includes transitions flipping up to $r$ bits on the hypercube, which typically improves the mixing time. 
Previous works~\cite{huang2025certifying,bravyi2022simulate} have investigated $\tau(\theta)$ for various states of physical interest. For examples, most Haar-random states have $\tau=\calO(n^2)$ when $r=1$; the ground state of a gapped $\kappa$-local stoquastic Hamiltonian has $\tau\leq\calO(n^{\kappa+1})$ for $r = \kappa$; phase states $\frac{1}{\sqrt{2^n}}\sum_{x\in\set{0,1}^n}e^{i\phi_x}\ket{x}$ have $\tau = n$ for $r=1$. In these cases, $\tau$ is a low-order polynomial in $n$, establishing the efficiency of our protocol. 

The protocol can be improved by incorporating adaptive measurements~\cite{gupta2025few}. The core idea is to leverage the adaptive \emph{decision tree} (DT) basis. For any pair of single-qubit states $\rho_1$ and $\rho_2$, there always exists a basis $\set{\ket{e_1},\ket{e_2}}$ such that both states have equal measurement probabilities $1/2$, i.e.,
\begin{align*}
    \bbraket{e_1|\rho_1|e_1}=\bbraket{e_1|\rho_2|e_1}=
    \bbraket{e_2|\rho_1|e_2}=\bbraket{e_2|\rho_2|e_2}=\frac{1}{2}.
\end{align*}
This follows by choosing $\set{\ket{e_1},\ket{e_2}}$ perpendicular to the Bloch-ball vector  representations of $\rho_1$ and $\rho_2$. 
Applying this construction inductively to all qubits, for any two states $\ket{\psi_{1,2}}$ one can construct an adaptive product-state basis $\set{\ket{\ell}}_\ell$ such that $\abs{\bbraket{\ell|\psi_1}}^2=\abs{\bbraket{\ell|\psi_2}}^2=1/2^n$ for all $\ell$~\cite{walgate2000local,zhou2020saturating,gupta2025few}. 
Leveraging this fact, Ref.~\cite{gupta2025few} proposed a certification protocol that achieves $\Delta = n$ for 
any pure target state, which we can adapt into a metrology protocol.
\begin{theorem}[Adaptive decision-tree protocol; informal]
\label{thm:main:2}
    Consider a POVM where we first choose a qubit uniformly randomly, say the $i$-th qubit. 
    We then measure the first $i-1$ qubits in the computational basis with outcome $x$, and the $i$-th qubit in the random Pauli basis, and the
    last $n-i$ qubits in a DT basis $\calT_x$ such that both $\bra{x_{}\otimes 0}\psi(\theta)\rangle/\norm{\bra{x\otimes 0}\psi(\theta)\rangle}_2$ and $\bra{x\otimes 1}\psi(\theta)\rangle/\norm{\bra{x\otimes 1}\psi(\theta)\rangle}_2$ are phase states.
    The CFI matrix for any pure target state $\ket{\psi(\theta)}$ satisfies
    \begin{equation}\label{eq:thm:main:2}
        I^{-1}(\calM,\ket{\psi(\theta)})\preceq 4n
        J^{-1}(\ket{\psi(\theta)}).
    \end{equation}
\end{theorem}\noindent
See Theorem~\ref{thm:metrology2} in Appendix~\ref{app:protocol2} for a more detailed discussion. 
This result can also be viewed as a generalization of Theorem~\ref{main:thm:1}, extending the case $\tau=n$ from phase states to arbitrary states.

The third metrology protocol improves the polynomial overhead to a constant overhead for generic Haar-random states using randomized Pauli measurement~\cite{du2025certifying,coladangelo2026robust}.
\begin{theorem}[Randomized Pauli measurement protocol; informal]
\label{thm:main:3}
    Consider a POVM where 
    we do randomized Pauli measurement on all the $n$ qubits.
    There exists a constant $C>0$ such that
    the CFI matrix for a Haar-random state $\ket{\psi(\theta)}$ satisfies
    \begin{equation}
        I^{-1}(\calM,\ket{\psi(\theta)})\preceq  C
        J^{-1}(\ket{\psi(\theta)}).
    \end{equation}
    with probability at least $1-e^{-\Omega(n)}$.
\end{theorem}\noindent
Note that randomized Pauli measurements are applied in both steps, thus we do not explicitly impose $r=1$ in the theorem description. See Theorem~\ref{thm:metrology4} in Appendix~\ref{app:haar1} for a more detailed discussion. There is an alternative protocol achieving a similar behavior that we describe in Appendix~\ref{app:haar2}, which combines a two-bases measurement on $n-r$ qubits with random Clifford measurement on $r$ qubits with $r = \Theta(\log(n))$~\cite{coladangelo2026power}.

\emph{Hamiltonian estimation example.---} 
We numerically demonstrate our protocol on the example of estimating Hamiltonian parameters from measurements on multiple copies of its gapped ground state. 
The Hamiltonian is an $n$-qubit transverse-field Ising model with longitudinal or transverse disorder,
\begin{align}\label{eq:ham}
    H = - \sum_{i=1}^{n-1}\sigma^z_i\sigma^z_{i+1}-g\sum_{i=1}^n \sigma^x_i+
    \begin{cases}
        \sum_{i=1}^n \theta^z_i\sigma^z_i; \\
        \sum_{i=1}^n \theta^x_i\sigma^x_i.
    \end{cases}
\end{align}
Throughout the simulation we fix $g=1.5$ to ensure a unique gapped ground state. The disorder strength per site is drawn from a Gaussian distribution with small variance, $\theta^z_i,\theta^x_i\overset{\text{i.i.d.}}{\sim}\calN(0,0.1)$. Given many copies of the gapped ground state $\ket{\psi(\theta^{x,z})}$, we use the POVM introduced in Theorem~\ref{main:thm:1} with $r=1$ to construct a metrology protocol for estimating $\theta^{x,z}$. Note that, because the Hamiltonian is stoquastic and 2-local, the theoretical bound yields $\tau\leq\calO(n^3)$ when $r=2$.

\begin{figure}[tb]
    \centering
    \includegraphics[width=1\linewidth]{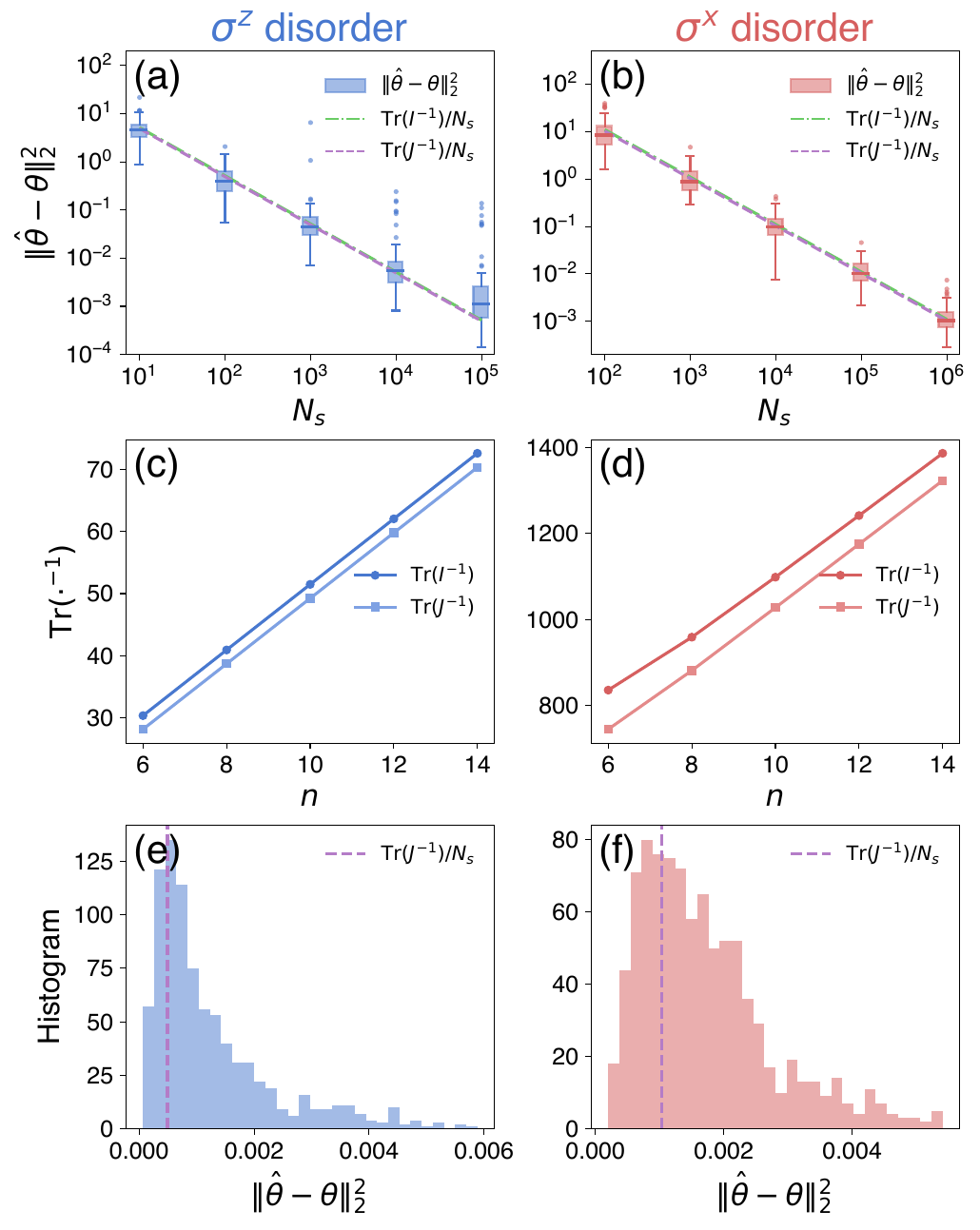}
    \caption{Parameter estimation for disordered transverse-field Ising model with $\sigma^z$ or $\sigma^x$ disorder, using the POVM from Theorem~\ref{main:thm:1} and the estimator from Eq.~\eqref{eq:local-mle}. (a)(b) 
    Total MSE versus measurement number $N_s$ for a typical random choice of $\theta^{z,x}$ at $n=10$. Box plots summarize the statistical distribution of the data over 100 independent experiments. The central line of the box represents the median, while the lower and upper edges indicate the first quartile and third quartile. Whiskers extend to the most extreme data points within 2 times the interquartile range from the quartiles. Individual points outside the whiskers are shown as outliers. 
    Numerical calculation shows $4\tau(\theta^0)\approx 113.8$ when $n = 10$, indicating $\tr(I^{-1})$ performs much better than its theoretical upper bound  $\tr(4\tau \cdot J^{-1})$ in this case.
    (c)(d) The scaling of traces of inverse CFI matrix and inverse QFI matrix versus system size $n$. Over this range of $n$, $10^2\lesssim 4\tau\lesssim 10^3$.
    (e)(f) Distribution of total MSE over 1000 random choices of $\theta^{x,z}$ when $n = 10$, for (e) $N_s=10^5$, and (f) $N_s=10^6$. We also exclude extreme data points outside the 2 times the interquartile range from the quartiles. Note that in all six subfigures, the CFI and QFI matrices are computed at $\theta^0$ as good approximations to their values at $\theta$.}
    \label{fig:numerics}
\end{figure}

To construct estimators from measurement outcomes, we employ the (leading-order) MLE near $\theta = \theta^0$, corresponding to the optimal locally unbiased estimator at $\theta = \theta^0$. Here $\theta^0$ is a known, prior estimate of $\theta$ which is close to the true value of $\theta$ and we take $\theta^0 = 0$ in this example. It corresponds to the situation where the disorder strengths are small. Given $N_s$ measurement outcomes $x_1,\cdots,x_{N_s}$, MLE minimizes the log-likelihood function $\calL(\theta) = -\sum_{\alpha=1}^{N_s}\log p_{x_\alpha}(\theta)$. Expanding it to first order around $\theta^0 = 0$, we have
{\small
\begin{align}
    &\;\;\partial_i\calL(\theta) = \sum_{\alpha=1}^{N_s} - \frac{\partial_i p_{x_\alpha}(\theta)}{p_{x_\alpha}(\theta)}\bigg|_{\theta=\theta^0} + \bigg( \frac{\partial_i p_{x_\alpha}(\theta) \partial_j p_{x_\alpha}(\theta)}{p_{x_\alpha}(\theta)^2}\nonumber\\ &-\frac{\partial_i  \partial_j p_{x_\alpha}(\theta)}{p_{x_\alpha}(\theta)} \bigg) \bigg|_{\theta=\theta^0} (\theta_j - \theta^0_j) + \calO(N_s \norm{\theta_j - \theta^0_j}^2) \\
    &\approx \sum_{\alpha=1}^{N_s} - \frac{\partial_i p_{x_\alpha}(\theta)}{p_{x_\alpha}(\theta)}\bigg|_{\theta=\theta^0}
    \!+ N_s \sum_{j=1}^m I_{ij}(\calM,\ket{\psi(\theta^0)})(\theta_j-\theta^0_j),\nonumber
\end{align}
}where we approximate the coefficients of the linear term by their expectation values at $\theta = \theta^0$ and $I(\calM,\ket{\psi(\theta^0)})$ is the CFI matrix. Setting each $\partial_i\calL(\theta)=0$ to minimize $\calL(\theta)$ and noting $\theta^0 = 0$ yield estimators
\begin{equation}
\label{eq:local-mle}
    \hat{\theta}_i(x_1,\ldots,x_{N_s}) := \frac{1}{N_s}\sum_{\alpha=1}^{N_s} \hat{\theta}_i(x_\alpha),
    \end{equation}
    where $\hat{\theta}_i(x_\alpha) := \sum_{j=1}^m (I^{-1}(\calM,\ket{\psi(\theta^0)}))_{ij}
    \frac{\partial_j p_{x_\alpha}(\theta)}{p_{x_\alpha}(\theta)}\big|_{\theta=\theta^0}$.
In our example, $p_x(\theta^0)=\norm{\bbraket{x|\psi(\theta^0)}}^2$, where $\ket{x}$ is the product state corresponding to measurement outcome $x$ and $\ket{\psi(\theta^0)}$ is the ground state of $H(\theta^0)$. 
Since $\theta^0 = 0$, $p_x(\theta^0)$, $\partial_ip_x(\theta^0)$, and $I_{ij}(\theta^0)$ are computed by applying the density matrix renormalization group (DMRG)~\cite{white1992density,white1993density,ostlund1995thermodynamic,schollwock2005density,perez2006matrix,verstraete2008matrix} method to the Hamiltonian (Eq.~\eqref{eq:ham}).

We illustrate the sample-dependent total MSE in estimating random $\sigma^z$ and $\sigma^x$ disorder strengths for $n=10$ in Fig.~\ref{fig:numerics}(a)(b). 
Here the sample-dependent total MSE is $\|\hat\theta(x_1,\ldots,x_{N_s})-\theta\|_2^2 = \frac{1}{N_s^2}\sum_{\alpha=1}^{N_s}\|\hat\theta(x_\alpha)-\theta\|_2^2$. $N_s$ times the sample-dependent total MSE is equal to the trace of the MSEM (Eq.~\eqref{eq:msem}), $\tr(V(\hat{\theta}(x)))$, with $p_x(\theta)$ replaced by its empirical distribution sampled from $N_s$ experiments, and it converges to $\tr(V(\hat{\theta}(x)))$ as $N_s$ approaches infinity.
From the plots, we observe the quantum CRB $\tr(V(\hat{\theta}(x_1,\ldots,x_{N_s}))\geq \tr(J^{-1})/N_s$ is nearly saturated by our estimators, indicating that the CFI matrix is within a small factor of the QFI matrix in this example. 
One can also numerically compute $\tau(\theta^0)$ in this case, yielding $4\tau(\theta^0)\approx 113.8$, meaning the actual performance is better than the theoretical bounds from Theorem~\ref{main:thm:1}.
We numerically simulate the CFI-QFI ratios  for various system sizes in Fig.~\ref{fig:numerics}(c)(d). Finally, in Fig.~\ref{fig:numerics}(e)(f) we  plot the error distribution over different instances of $\theta^{z,x}$ for fixed values of $N_s$. The distribution is strongly concentrated around the inverse-QFI prediction, confirming the universality of the behavior shown in Fig.~\ref{fig:numerics}(a)(b).
It can also be seen that the MSE for $\sigma^z$ disorder is significantly smaller than that for $\sigma^x$ disorder (for the same $N_s$). This is because in the paramagnetic phase, the ground state has a large overlap with $\ket{+}^{\otimes n}$ and adding $\sigma^x$ perturbations does not change it to leading order. Thus, less information about the perturbation strength can be obtained from ground states.

Note that Fig.~\ref{fig:numerics}(c)(d) also demonstrate a linear dependence of $\trace{}{I^{-1}}$ with system size, roughly yielding a sample complexity $N_s\sim n$ to achieve a fixed error. 
On the other hand, the state-of-the-art heuristic quantum algorithm requires constructing a full correlation matrix for local observables~\cite{qi2019determining}, which takes at least $\sim n$ copies when taking into account the sparsity~\cite{bairey2019learning}. Our results indicate this scaling could be optimal. Our protocol also achieves this scaling using simple single-qubit measurements.



\emph{Summary and outlook.---}We introduce a family of quantum metrology protocols that use only few-qubit measurements, substantially reducing resource requirements. Despite this restriction, we prove that they 
approach the optimal MSEM set by the quantum Cramér–Rao bound up to only low-polynomial overhead. We illustrate the approach through Hamiltonian estimation from ground states using single-qubit measurements, with numerical results suggesting that performance can exceed our theoretical guarantee in some instances.

We also establish a universal connection between quantum metrology and CFE-based quantum state certification, enabling us to derive these performance bounds. Further exploring this connection may deepen our understanding of both fields, inspire algorithms in both directions, and broaden their applications.


Our results leave several important open questions. First, our analysis is limited to pure states. Although recent work on random purification channels~\cite{chen2024local,tang2025conjugate,girardi2025random,walter2025random,pelecanos2025mixed,zhou2026quantum} provides a method to bridge mixed and pure states, the transformation can be costly. Whether generic mixed-state metrology can be performed efficiently using few-qubit measurements remains open.
Second, incorporating measurement noise or quantum error correction would bring the framework closer to practical implementation. Finally, our protocols operate in the high-precision regime, assuming prior, close estimates of $\theta$. Extension from local estimation to global learning remains largely unexplored~\cite{kwon2026universal,chen2026instance}. 

\medskip

\emph{Note.---}During the preparation of this manuscript, we became aware of Ref.~\cite{du2026complexity}, which proposes a 
measurement strategy for multi-parameter metrology with polylogarithmic circuit complexity. Our work uses a different approach with simpler few-qubit measurements. 

\emph{Acknowledgment.---} We are grateful to Zhenyu Du for sharing a recent result on Haar-random states certification~\cite{du2025certifying}. Part of the work was done when L.M. was visiting California Institute for Technology. S.C., H.-Y.H., J.P.  acknowledge funding provided by the Institute for Quantum Information and Matter, an NSF Physics Frontiers Center (NSF Grant PHY-2317110), and the U.S. Department of Energy, Office of Science, National Quantum Information Science Research Centers, Quantum Systems Accelerator.  S.Z. acknowledges funding provided by Perimeter Institute for Theoretical Physics, a research institute supported in part by the Government of Canada through the Department of Innovation, Science and Economic Development Canada and by the Province of Ontario through the Ministry of Colleges and Universities.

\bibliographystyle{unsrt}
\bibliography{ref}

\onecolumngrid

\newpage

\onecolumngrid
\appendix

\numberwithin{equation}{section}
\setcounter{theorem}{0}
\setcounter{figure}{0}
\renewcommand{\thefigure}{S\arabic{figure}}
\renewcommand{\thelemma}{S\arabic{lemma}}
\renewcommand{\thetheorem}{S\arabic{theorem}}
\renewcommand{\thecorollary}{S\arabic{corollary}}
\renewcommand{\theHfigure}{Supplement.\arabic{figure}}
\renewcommand{\theHtheorem}{Supplement.\arabic{theorem}}

\tableofcontents

\section{Preliminaries on quantum metrology}\label{app:review}

Consider a $d$-dimensional parameterized quantum state $\rho(\theta)$ in Hilbert space $\mH$, where parameters are denoted by $\theta = (\theta_1,\theta_2,\ldots,\theta_m)$. $m$ is the number of parameters and $\Theta \subseteq \bR^m$ is the domain of $\theta$. An estimator $\htheta(x)$ and the corresponding positive operator-valued measurement (POVM) $\calM = \{M_x\}_{x}$ is a function that maps the measurement outcomes $x$ to $\Theta$. Here $M_x$ are positive semidefinite measurement operators satisfying $\sum_x M_x = \id$.  The performance of an estimator at $\theta$ is characterized by the \emph{mean squared error matrix} (MSEM)
\begin{equation}
    V(\calM,\htheta)_{ij} = \sum_x (\htheta(x)_i - \theta_i)(\htheta(x)_j - \theta_j) \ttrace(\rho(\theta) M_x). 
\end{equation}
For single-parameter estimation (i.e., $m = 1$), we call $V$ the mean squared error (MSE). 

In the paradigm of local parameter estimation, we focus on analyzing the performance of \emph{locally unbiased estimators} that are unbiased at the true value of $\theta$ and in its vicinity up to first order. It is relevant in metrological situations where the prior knowledge of $\theta$ is already given in the vicinity of its true value, or roughly obtained from a pre-estimation phase (which usually consumes a number of experiments that are negligible compared to $N$, the total number of copies of states available)~\cite{rao1973linear,kay1993fundamentals,lehmann2006theory,cox2017inference,casella2002statistical,gill2000state,yang2019attaining}. Precisely, an estimator $(\htheta,\calM)$ is locally unbiased at $\theta = \theta^0$ if and only if 
\begin{gather}
\label{eq:l-u-condition-1}
    \theta^0 = \sum_{x} \htheta(x) \ttrace(\rho(\theta) M_x) \Big|_{\theta = \theta^0},\\
\label{eq:l-u-condition-2}
    \delta_{ij} = \frac{\partial}{\partial \theta_i}\sum_{x} \htheta_j(x) \ttrace(\rho(\theta) M_x) \Big|_{\theta = \theta^0},\;\forall i,j. 
\end{gather}  
In classical statistics, the MSEM of any locally unbiased estimator at $\theta$ is bounded below by the inverse of the CFI matrix, defined by 
\begin{equation}
    I(\calM)_{ij} = \sum_{x, p_x:=\ttrace(\rho(\theta) M_x) \neq 0} \frac{1}{p_x(\theta)} \frac{\partial p_x(\theta)}{\partial \theta_i} \frac{\partial p_x(\theta)}{\partial \theta_j}. 
\end{equation}
$I(\calM)$ is a function of the POVM and is independent of the estimators. $I(\calM)$ is a positive semidefinite matrix in $\bR^{m\times m}$ and we further assume in this work that $I(\calM)$ is strictly positive, which means each parameter is identifiable (unless stated otherwise). 
This forbids redundant parameterization and puts an upper bound on the number of allowed parameters.
Note that $I(\calM)$ depends on $\theta$. In below we use $I(\calM,\rho)$ (or $I(\calM,\ket{\psi})$ for pure state) when the target state, and hence $\theta$, is specified.
The Cram\'{e}r--Rao bound (CRB)~\cite{rao1973linear,kay1993fundamentals,lehmann2006theory,cox2017inference,casella2002statistical} states that for any locally unbiased estimator, 
\begin{equation}
\label{eq:CRB}
    V(\calM,\htheta) \succeq I(\calM)^{-1},
\end{equation}
where $A \succeq B$ denotes a relationship between two positive semidefinite matrices $A$ and $B$ where $A - B$ is positive semidefinite. The CRB at $\theta^0$ is saturable via the following locally unbiased estimator~\cite{rao1973linear,kay1993fundamentals,lehmann2006theory,cox2017inference,casella2002statistical} defined by 
\begin{equation}
\label{eq:opt}
    \htheta^{\opt}_i(x;\theta^0) = \theta_i^0 +   \sum_{i} I(M)^{-1}_{ij} \frac{\partial_j p_x}{p_x} \bigg|_{\theta = \theta^0},
\end{equation}
where $(~; \star)$ indicates dependence on $\star$ and $\partial_i$ is used as a shorthand for $\frac{\partial}{\partial \theta_i}$. $\bE[\hat{p}_x] = p_x = \ttrace(\rho(\theta) M_x)$ and $\bE[\hat{p}_x \hat{p}_y] =\delta_{xy} p_x$. It achieves the CRB, i.e., 
\begin{equation}
    V(\calM,\hat\theta^\opt) = I(\calM)^{-1}. 
\end{equation}
Given $N$ copies of $\rho(\theta)$, the CRB is asymptotically saturable for large $N$ using the (asymptotically unbiased) maximum likelihood estimator $\htheta^{(N)}_{\textsc{mle}}$~\cite{rao1973linear,kay1993fundamentals,lehmann2006theory,cox2017inference,casella2002statistical}---$\sqrt{N}(\hat{\theta}^{(N)}_{\textsc{MLE}} - \theta)$ converges in distribution to a normal distribution centered around $0$ with variance equal to $I(\calM)^{-1}$ as $N \rightarrow \infty$. Note that $I(\calM)$ here is additive---$I(\calM^{\otimes N}) = N I(\calM)$, where we use $I(\calM^{\otimes N})$ to denote the CFI matrix when measuring $\rho(\theta)^{\otimes N}$ using POVM $\calM^{\otimes N} = \{\bigotimes_{i=1}^N M_{x_i}\}_{(x_1,x_2,\ldots,x_N)}$. The factor of $N$ naturally occurs from the classical central limit theorem, where the estimation variance is improved by a factor of $N$ when $N$ i.i.d. samples are taken. 

The CRB above is a purely classical result. It provides a lower bound on $V(\calM,\htheta)$ that depends on the choice of POVM. The quantum Cram\'{e}r--Rao bound (QCRB)~\cite{helstrom1967minimum,helstrom1968minimum,helstrom1976quantum,holevo2011probabilistic,braunstein1994statistical,barndorff2000fisher,paris2009quantum} provides a more general lower bound when the POVM can be chosen freely. For any locally unbiased estimator for target state $\rho$, 
\begin{equation}
\label{eq:QCRB-app}
    V(\calM,\htheta) \succeq I(\calM,\rho)^{-1} \succeq J^{-1}(\rho),
\end{equation}
where $J$ is the QFI matrix as a function of $\rho(\theta)$. $J$ is also additive with respect to $\rho(\theta)$. It is defined by 
\begin{equation}
    J_{ij} = \frac{1}{2}\ttrace(\rho(\theta) \{L_i,L_j\}),
\end{equation}
where $\{\cdot,\cdot\}$ is the anti-commutator and $L_i$ are Hermitian operators called symmetric logarithmic derivative (SLD) operators, defined through 
\begin{equation}
    \frac{\partial \rho(\theta)}{\partial \theta_i} = \frac{1}{2}(L_i \rho(\theta) + \rho(\theta) L_i). 
\end{equation}
$L_i$ is not uniquely defined by the equation above, but $J$ is invariant under different choices of $L_i$. In this work, we focus on the cases where $\rho=\ketbrat{\psi}$ is a pure state. Under this condition,
\begin{align}
    L_i = 2(\partial_i\ketbrat{\psi}+\ket{\psi}\bra{\partial_i\psi}).
\end{align}
The QFI matrix can be calculated as
\begin{align}
    J_{ij} = \frac{1}{2}\trace{}{L_iL_j}=4\operatorname{Re}\left[
    \braket{\partial_i\psi}{\partial_j\psi}-\braket{\psi}{\partial_i\psi}\braket{\partial_j\psi}{\psi}
    \right]
\end{align}

For single-parameter estimation, i.e., when $m = 1$, $J$ and $I(\calM)$ are scalars, and $J$ is attainable~\cite{braunstein1994statistical}:
\begin{equation}
\label{eq:single}
    J = \max_{\calM:\mathrm{POVM}}I(\calM).
\end{equation}
Therefore, the QCRB (\eqref{eq:QCRB-app}) is asymptotically attainable, making the QFI $J$ a perfect figure of merit to characterize the performance of local parameter estimation on state $\rho(\theta)$. 
For multi-parameter estimation, however, the QFI matrix $J$ is not always attainable by maximizing $I(\calM)$ over $\calM$. The reason is the measurements that are optimal with respect to different parameters may not be compatible with each other, i.e., cannot be implemented simultaneously. 

The Gill--Massar (GM) bound~\cite{gill2000state} describes a relation between the QFI matrix and CFI matrix that applies to all individual measurements, i.e. measurements that act on each copy of quantum states individually. It states that for any $d$-dimensional quantum state $\rho(\theta)$ and any (individual) POVM $\tilde{\calM}$,
\begin{equation}
\label{eq:GM}
    \ttrace(J^{-1}I(\tilde{\calM})) \leq d-1. 
\end{equation}
It sets a lower bound on $\ttrace(J I(\tilde{\calM})^{-1})$ through Cauchy--Schwarz inequality, 
\begin{equation}
\label{eq:lower}
    \ttrace(J I(\tilde{\calM})^{-1}) \geq \frac{\ttrace(\id)^2}{\ttrace(J^{-1}I(\tilde{\calM}))} \geq \frac{m^2}{d-1}. 
\end{equation}
Given a parameterized pure state $\rho(\theta)$, it was known that there exists a rank-one measurement $\calM$ (with $2d-1$ measurement outcomes) such that the Fisher-symmetry condition is satisfied~\cite{li2016fisher}, i.e., 
\begin{equation}
\label{eq:f-s-pure}
    I(\calM) = \frac{1}{2} J.
\end{equation}
In fact, $1/2$ is the best possible constant for state tomography, i.e., when the number of parameters $m = 2(d-1)$. One can see this from the GM bound---if $I(\calM) = c J$ from some constant $c$, then $
\ttrace(J^{-1} I(\calM)) = c m = 2c(d-1) \leq d-1$. Fisher-symmetric measurements are defined to be measurements that achieves $I(\calM) \propto J$ with the optimal coefficient. 

The above measurement~\cite{li2016fisher} is a function of the state $\rho(\theta)$ and in practical implementation, must be adjusted adaptively based on the accumulated knowledge of the states~\cite{vargas2024near}. In addition, both the computational complexity and implementation complexity of this approach can be costly. 
To overcome the limitation on state-dependence, researchers studied universally Fisher-symmetric measurements that are independent of states. It was known the infinite-outcome Haar random measurements~\cite{hayashi1998asymptotic,zhu2014quantum,lu2025quantum} are universally Fisher-symmetric and can achieve \eqref{eq:f-s-pure} for all pure states universally. However, they are experimentally infeasible because they require infinite amount of classical randomness and exponential gate complexity. 
In fact, \cite{zhu2018universally} showed any measurement $\calM$ satisfying \eqref{eq:f-s-pure} for all pure states must contains infinite number of measurement outcomes (and thus are not feasible). They proposed the $2$-design collective measurements on two identical copies of the state as a workaround (as also experimentally demonstrated for the single-qubit case~\cite{hou2018deterministic}). They showed that 
\begin{equation}
    I(\calM,\rho(\theta)^{\otimes 2}) = J(\rho(\theta)^{\otimes 2}),
\end{equation}
when $\calM$ is a $2$-design measurement on $2$ copies of any pure state $\rho(\theta)$---which surpasses the GM bound because the measurement is no longer single-copy.  In practice, implementing a two-copy measurement requires simultaneous access to and precise cross-device control of two identical quantum sensors. A more experimentally accessible single-copy (individual) measurement protocol was proposed recently~\cite{zhou2026randomized} which applies 3-design measurements on individual copies of pure states. Although not Fisher-symmetric, the measurement protocol can achieve the optimal QFI matrix up to a constant factor. In fact, when $\calM$ is 3-design, it was shown 
\begin{equation}
    I(\calM) \succeq \frac{d+2}{4(d+1)} J .
\end{equation}
Meanwhile, locally unbiased estimators called local shadow estimators were constructed in~\cite{zhou2026randomized} to achieve the right-hand side of the inequality above, where the estimators are evaluated as the expectation values of a set of observables---called deviation observables---on the pure state, estimated using classical shadows~\cite{huang2020predicting}. Still, 3-design measurements are entangled measurements between qubits and take polynomial-size quantum circuits to implement. The experimental challenge still exists in near-term physical platforms.

\section{The conditional fidelity estimation protocols}\label{app:certification}

\subsection{The conditional fidelity estimation protocols for certification}
Another central task in quantum information science is \emph{quantum state certification}, where the goal is to test whether a physical state $\rho$ prepared in the lab is close to a specific target state $\ket{\psi}$. Such a fundamental task has applications from quantum device benchmarking to verification of quantum algorithms.

More precisely, the setting of certification problem can be described as below: Given many copies of an unknown lab state $\rho$ and a classical description of the target state $\ket{\psi}$, the goal is to determine whether the fidelity $f(\rho,\ket{\psi}):=\bra{\psi}\rho\ket{\psi}$ is greater than $1-\epsilon$ or not, for a given accuracy parameter $\epsilon$. A trivial protocol to this problem is to measure a POVM $\calM=\set{\ketbrat{\psi},I-\ketbrat{\psi}}$. However, this POVM demand one to implement the state preparation circuit of $\ket{\psi}$, which is typically of high complexity, if possible. On the other hand, one may use classical-shadow-based protocols to measure fidelity, which however still demand accurate implementations of a large number of multi-qubit gates, exceeding the capability of current quantum devices

Recently, an alternative family of protocols have emerged and been proved to be efficient. Those protocols ultilize only single or few qubit measurements, i.e., implementing unitaries on few concatenated qubits in parallel before measuring the computational basis.  When the unitaries are single-qubit ones, the protocol only requires single-qubit (possibly randomized) measurements. The price of using those very low-complexity operations is that, a \emph{certification gap} $\Delta>1$ exists in this family of protocols. That is, for a given accuracy parameter $\epsilon$, one is only able to correctly output $\certified$ or $\failed$ (with high probability) when $f(\rho,\ket{\psi})>1-\epsilon$ or $f(\rho,\ket{\psi})<1-\Delta \epsilon$. 

\begin{algorithm}[t!]
\caption{Conditional fidelity estimation protocol for certification}
\label{alg:cfe-certification}
\KwIn{A lab state $\rho$, classical description of a target state $\ket{\psi}$.}
\KwOut{$\certified$ or $\failed$}
     Randomly select a POVM $\calM_\alpha\in\calM$ according to probability $p_\alpha$.\\
     Measure $\rho$ to get a measurement outcome $x|i$ and a post-measurement state on $r$ qubits, 
     \begin{equation*}
         \rho_{x|\alpha}=\frac{K^{(\alpha)}_x\rho K^{(\alpha)\dagger}_x}{\trace{}{K^{(\alpha)}_x\rho K^{(\alpha)\dagger}_x}}.
     \end{equation*}
     \\
     Classically compute the post-measurement state of $\ket{\psi}$ after the same measurement,
     \begin{equation*}
         \ket{\psi_{x|\alpha}}=\frac{K^{(\alpha)}_x\ket{\psi}}{\sqrt{\trace{}{K^{(\alpha)}_x\ketbrat{\psi} K^{(\alpha)\dagger}_x}}}.
     \end{equation*}\\
     Measure $\rho_{x|\alpha}$ with the POVM
     \begin{equation*}
         \set{\ketbrat{\psi_{x|\alpha}},I-\ketbrat{\psi_{x|\alpha}}}.
     \end{equation*}
     If the outcome is $\ketbrat{\psi_{x|\alpha}}$, output $\certified$, otherwise output $\failed$.
\end{algorithm}

A nice feature of this family of protocols, which enables us to build connections to metrology, is that they are based on fidelity estimation of some small subsystem conditioned on the measurement outcomes of the rest. For this reason, we call those protcols the conditional fidelity estimation (CFE) protocols. We summarize the essential features below. 
\begin{definition}
    [Conditional fidelity estimation protocols for certification]
    \label{def:cfe-certification}
    A conditional fidelity estimation protocol for certifying an $n$-qubit state is specified by a tuple $\{ \mathcal{M}, r, \Delta \}$,
where:
\begin{itemize}
    \item $\mathcal{M} = \{ \mathcal{M}_\alpha, p_\alpha \}_\alpha$ is a family of randomized measurements indexed by a classical parameter $\alpha$, sampled according to probability $p_\alpha$;
    \item $r \in \mathbb{N}^+$ is the number of qubits left unmeasured (the post-measurement subsystem);
    \item $\Delta > 1$ is the certification gap.
\end{itemize}

Each measurement $\mathcal{M}_\alpha$ is specified by:
\begin{itemize}
    \item a subset $S_\alpha \subseteq [n]$ of size $n - r$ indicating the qubits to be measured;
    \item a POVM  described by Kraus operators $\set{K_{x}^{(\alpha)}}_x$ acting on the qubits in $S_\alpha$.
\end{itemize}
The corresponding Kraus operators take the form
$K_{x}^{(\alpha)} = \widetilde{K}_{\alpha,x}^{(S_\alpha)} \otimes \id_{S_\alpha^c}$,
where $\widetilde{K}_{\alpha,x}^{(S_\alpha)}$ acts on the measured qubits $S_\alpha$, and $\id_{S_\alpha^c}$ is the identity operator on the remaining $r$ qubits.
The index $\alpha$ specifies the choice of measured subset $S_\alpha$ and the measurement basis, while $x$ denotes the classical measurement outcome according to Born's rule.

    The procedure of the protocol is described by Algorithm.~\ref{alg:cfe-certification}. The algorithm satisfies the following requirements:
    \begin{itemize}
        \item Completeness. --- The probability to output $\certified$ is at least $\bra{\psi}\rho\ket{\psi}$.
        \item Soundness. --- The probability to output $\failed$ is at least 
        $(1-\bra{\psi}\rho\ket{\psi})/\Delta$.
    \end{itemize}
\end{definition}

We note that, all three parameters of a CFE protocol, $\calM$, $r$ and $\Delta$, may depend on target state $\ket{\psi}$. Here we slightly abuse the notation to omit the $\ket{\psi}$ dependence, and only mention it when necessary.

\begin{algorithm}[t!]
\caption{Multi-copy conditional fidelity estimation protocol for certification}
\label{alg:multi-cfe-certification}
\KwIn{$k$ copies of lab state $\rho$, classical description of a target state $\ket{\psi}$, accuracy parameter $\epsilon\in(0,1)$ and failure probability $\delta\in(0,1)$.}
\KwOut{$\certified$ or $\failed$}
\For{ $1\leq t\leq k$ copies of $\rho$}{
     Randomly select a POVM $\calM_\alpha\in\calM$ according to the probability $p_\alpha$.\\
     Measure $\rho$ to get a measurement outcome $x|i$ and a post-measurement state on $r$ qubits, 
     \begin{equation*}
         \rho_{x|\alpha}=\frac{K^{(\alpha)}_x\rho K^{(\alpha)\dagger}_x}{\trace{}{K^{(\alpha)}_x\rho K^{(\alpha)\dagger}_x}}.
     \end{equation*}
     \\
     Apply a $r$-qubit random Clifford circuit $U$ to $\rho_{x|\alpha}$ and then measure computational basis. The measurement outcome is $z\in\set{0,1}^r$.
     \\
     Classically compute the post-measurement state of $\ket{\psi}$ after the same measurement,
     \begin{equation*}
         \ket{\psi_{x|\alpha}}=\frac{K^{(\alpha)}_x\ket{\psi}}{\sqrt{\trace{}{K^{(\alpha)}_x\ketbrat{\psi} K^{(\alpha)\dagger}_x}}}.
     \end{equation*}\\
     Compute the classical shadow estimator for fidelity $\bra{\psi_{x|\alpha}}\rho_{x|\alpha}\ket{\psi_{x|\alpha}}$:
     \begin{equation*}
         f_t = (2^r+1)\abs{\bra{z}U\ket{\psi_{x|\alpha}}}^2-1.
     \end{equation*}
     }
     Compute the empirical value 
     \begin{equation*}
         f = \frac{1}{k}\sum_{t=1}^kf_t.
     \end{equation*}
     If $f_t>1-3\epsilon/(4\Delta)$, output $\certified$, otherwise output $\failed$.
\end{algorithm}

 Algorithm.~\ref{alg:cfe-certification} gives the single-copy version of the certification protocol. Using standard Chernoff bound and classical shadow tomography, it is easy to extend it to a multi-copy version with random Clifford measurements in the second stage of the measurement, which is more common in the literature. Here we provide a proof for the sake of completeness.
\begin{lemma}
    [Multi-copy conditional fidelity estimation protocols for certification]
    \label{lem:multi-cfe-certification}
    The multi-copy version of certification protocols based on conditional fidelity estimation, Algorithm.~\ref{alg:multi-cfe-certification}, satisfies:
    \begin{itemize}
        \item  If $\bra{\psi}\rho\ket{\psi}>1-\epsilon/(2\Delta)$, output $\certified$ with probability at least $1-\delta$;
        \item  If $\bra{\psi}\rho\ket{\psi}<1-\epsilon$, output $\failed$ with probability at least $1-\delta$,
    \end{itemize}
    using number of copies
    \begin{equation}
        k = \calO\left(2^r\frac{\Delta^2}{\epsilon^2}\log\frac{1}{\delta}\right).
    \end{equation}
\end{lemma}
\begin{proof}
    Using the 3-design properties of random Clifford circuit, we have
    \begin{gather}
        \mathbb{E}_{U } \left[ U^\dagger |x\rangle \langle x| U \, \langle x| U A U^\dagger |x\rangle\right]
= \frac{A + \operatorname{tr}(A)\id}{(2^n + 1) 2^n}
    \end{gather}
    So the empirical average converges to
    \begin{align}
        \Exp{\alpha,x,U,z}{f_t} &= \sum_\alpha p_\alpha\sum_x\trace{}{K^{(\alpha)}_x\rho K^{(\alpha)\dagger}_x}\sum_{z\in\set{0,1}^r}
        \Exp{U}{\bra{z}U^\dagger\rho_{x|\alpha}U\ket{z}\left((2^r+1)\abs{\bra{z}U\ket{\psi_{x|\alpha}}}^2-1\right)}
        \notag\\
        &=\sum_\alpha p_\alpha\sum_x\trace{}{K^{(\alpha)}_x\rho K^{(\alpha)\dagger}_x}\bra{\psi_{x|\alpha}}\rho_{x|\alpha}\ket{\psi_{x|\alpha}}
    \end{align}
    Note that $|f_t|\leq 2^r$. Using Chernoff bound, we have
    \begin{equation}
        \mathrm{Pr}\left(\left|
        \frac{1}{k}\sum_{t=1}^kf_t-\Exp{\alpha,x,U,z}{f_t}\right|\geq\epsilon
        \right)\leq 2e^{-\frac{k\epsilon^2}{2^r}}.
    \end{equation}

    Now we choose 
    \begin{equation}
        k = 2^{r+4}\frac{\Delta^2}{\epsilon^2}\log\frac{2}{\delta}.
    \end{equation}
    Then 
    \begin{equation}
        \mathrm{Pr}\left(\left|
        \frac{1}{k}\sum_{t=1}^kf_t-\Exp{\alpha,x,U,z}{f_t}\right|\geq\frac{\epsilon}{4\Delta}
        \right)\leq \delta
    \end{equation}
    As a result, when $\bra{\psi}\rho\ket{\psi}>1-\epsilon/(2\Delta)$, from the completeness condition of Definition.~\ref{def:cfe-certification}, $\Exp{\alpha,x,U,z}{f_t}\geq 1-\epsilon/(2\Delta)$. With probability at least $1-\delta$, $(\sum_{t=1}^kf_t)/k>1-3\epsilon/(4\Delta)$ and Alg.~\ref{alg:multi-cfe-certification} outputs $\certified$. When $\bra{\psi}\rho\ket{\psi}<1-\epsilon$, 
    \begin{align}
        \Exp{\alpha,x,U,z}{f_t}\leq 1-\frac{1-\bra{\psi}\rho\ket{\psi}}{\Delta}<1-\frac{\epsilon}{\Delta}.
    \end{align}
    With probability at least $1-\delta$, $(\sum_{t=1}^kf_t)/k<1-3\epsilon/4\Delta$ and Alg.~\ref{alg:multi-cfe-certification} outputs $\failed$.
\end{proof}

For all protocols discussed later, $r$ is at most $\calO(\log n)$. So Alg.~\ref{alg:multi-cfe-certification} is efficient in terms of sample complexity.

In certification protocols, we assume the query access to the amplitudes of target state in one or multiple sets of product bases. In practice, when the target wavefunction has a closed analytical formula, or can be expressed as a matrix product state (MPS), we naturally have this access. We will assume a similar access model when transforming the certification protocols to those for metrology: query access to the amplitudes in the same set of bases as functions of $\theta$ at the neighborhood of $\theta^0$ (a prior, close estimate of $\theta$).

\subsection{Transforming certification protocols to metrology protocols}
Here we prove the key technical lemma. We show that every CFE-based certification protocol can be transformed into a metrology protocol with a provable CFI guarantee.
Moreover, we will prove this result be proving a stronger statement, i.e., 
we explicitly construct the locally unbiased estimators associated with the metrology protocol.

\begin{lemma}
    [Transforming certification protocols to metrology protocols]
    \label{lem:main}
    For every CFE-based certification protocol labeled by $\set{\calM,r,\Delta}$ defined in definition.~\ref{def:cfe-certification} for target state $\ket{\psi}$, there exists a corresponding metrology protocol that uses POVM $\calM'$ and achieves the CFI
    \begin{equation}
        I^{-1}(\calM',\ket{\psi})\preceq 4\Delta J^{-1}(\ket{\psi}).
    \end{equation}
    Moreover, $\calM'$ is obtained by measuring the classical shadow of the $r$ unmeasured qubits of $\calM$.
\end{lemma}
\begin{proof}
We prove this result by explicitly construct the set of locally unbiased estimators from $\calM'$. 
Let $\theta_i$ denotes independent parameters and $(\theta^0)_i:=\theta_i^0$, then the locally unbiased estimators of the metrology protocol are
    \begin{equation}
        \hat\theta_i((x|\alpha),(U,z)) = \theta_i^0 + \trace{}{X_i^{(x|\alpha)} \hat\rho(U,z)^{(x|\alpha)}}.
    \end{equation}
    $\rho(U,z)^{(x|\alpha)}$ is the classical shadow $(2^r+1)U^\dagger\ketbrat{z}U-\id$ of post-measurement state 
    \begin{equation}
        \ket{\psi_{x|\alpha}(\theta)}=\frac{K^{(\alpha)}_x\ket{\psi(\theta)}}{\sqrt{\trace{}{K^{(\alpha)}_x\ketbrat{\psi(\theta)} K^{(\alpha)\dagger}_x}}},
    \end{equation}
    obtained by applying random Clifford unitary $U$ on $\rho_{x|\alpha}$ and measuring the computational basis on $r$ qubits. $X_i^{(x|\alpha)}$ is the weighted SLD operator
    \begin{gather}
        X_i^{(x|\alpha)} = \sum_{j} ( J^{-1}_{M} )_{ij} L_j^{(x|\alpha)} \Big|_{\theta = \theta^0},\quad
        L_i^{(x|\alpha)} = 2\left(
        \ketbra{\partial_{\theta_i}\psi_{x|\alpha}(\theta)}{\psi_{x|\alpha}(\theta)}+\mathrm{h.c.}
        \right),
    \end{gather}
    where
    \begin{equation}
        J_M = \Exp{\alpha,x}{J\left(\ketbrat{\psi_{x|\alpha}(\theta)}\right) }=
        \sum_{\alpha,x|\alpha} p_{\alpha}\trace{}{K^{(\alpha)}_x\ketbrat{\psi(\theta)} K^{(\alpha)\dagger}_x}
        J\left(\ketbrat{\psi_{x|\alpha}(\theta)}\right) \,\Big|_{\theta = \theta^0},
    \end{equation}
    defines the \emph{averaged QFI matrix}. 

    It suffices to prove the covariance matrix $V$ for this protocol satisfies
    \begin{equation*}
         V\preceq 4\Delta J^{-1}.
    \end{equation*}
    We take two steps. First, we show the set of estimators we construct is indeed locally unbiased. Next, we prove the relation of $V$ to QFI.

    \textbf{Locally unbiased estimator}.--- We verify directly,
    \begin{align}
        \bbE_{(x|\alpha),(U,z)}\left[\hat\theta_i((x|\alpha),(U,z))\right] \Big|_{\theta = \theta^0} &= \theta_i^0 + \sum_{(x|\alpha)}p_{\alpha}\trace{}{K^{(\alpha)}_x\ketbrat{\psi(\theta)} K^{(\alpha)\dagger}_x}
        \trace{}{X_i^{(x|\alpha)} \Exp{(U,z)}{\hat\rho(U,z)^{(x|\alpha)}}}\Big|_{\theta=\theta^0} \notag\\
        &= \theta_i^0 + \sum_{(x|\alpha)}p_{\alpha}\trace{}{K^{(\alpha)}_x\ketbrat{\psi(\theta)} K^{(\alpha)\dagger}_x}
         \trace{}{X_i^{(x|\alpha)} \ketbrat{\psi_{x|\alpha}(\theta)}}\Big|_{\theta=\theta^0} \notag\\
         &=\theta_i^0
    \end{align}
    Here we use the fact that $\braket{\partial_{\theta_i}\psi(\theta)}{\psi(\theta)}=0$, because $\ket{\partial_{\theta_i}\psi(\theta)}$ must be in an orthogonal subspace with $\ket{\psi(\theta)}$, since both are normalized.
    Likewise,
    \begin{align}
\partial_{\theta_j} \Exp{(k,x),s}{\hat\theta_i((x|\alpha),(U,z))} \Big|_{\theta = \theta^0} &=  
\sum_{(x|\alpha)}p_{\alpha}\trace{}{K^{(\alpha)}_x\partial_{\theta_j}\left(\ketbrat{\psi(\theta)} \right)K^{(\alpha)\dagger}_x}
        \trace{}{X_i^{(x|\alpha)} \Exp{(U,z)}{\hat\rho(U,z)^{(x|\alpha)}}}\Big|_{\theta=\theta^0} 
        \notag\\
        &\quad +\sum_{(x|\alpha)}p_{\alpha}\trace{}{K^{(\alpha)}_x\ketbrat{\psi(\theta)} K^{(\alpha)\dagger}_x}
        \trace{}{X_i^{(x|\alpha)} \partial_{\theta_j}\Exp{(U,z)}{\hat\rho(U,z)^{(x|\alpha)}}}\Big|_{\theta=\theta^0} 
        \notag\\
        &=\sum_{(x|\alpha)}p_{\alpha}p(x|\alpha,\theta)
        \sum_{k} ( J^{-1}_{M} )_{ik} \trace{}{L_k^{(x|\alpha)} \cdot
        \frac{1}{2}L^{(x|\alpha)}_j
        }\Big|_{\theta=\theta^0}\notag\\
        &=\sum_{k}( J^{-1}_{M} )_{ik}\sum_{(x|\alpha)}p_{\alpha}p(x|\alpha,\theta)
        J_{kj}\left(\ketbrat{\psi_{x|\alpha}(\theta)}\right)\Big|_{\theta=\theta^0}\notag\\
        &=\sum_{k}( J^{-1}_{M} )_{ik} (J_M)_{kj}\notag\\
        &=\delta_{ij}
\end{align}
On the second equality we use $\braket{\partial_{\theta_i}\psi(\theta)}{\psi(\theta)}=0$, define $p(x|\alpha,\theta)=\trace{}{K^{(\alpha)}_x\ketbrat{\psi(\theta)} K^{(\alpha)\dagger}_x}$, and use $\Exp{(U,z)}{\hat\rho(U,z)^{(x|\alpha)}}=\ketbrat{\psi_{x|\alpha}(\theta)}$. On the third equality, we use that $J_{ij}(\ketbrat{\psi})|_{\theta=\theta^0}=\frac{1}{2}\trace{}{L_iL_j}|_{\theta=\theta^0}$. On the fourth equality, we use definition of $J_M$.

\textbf{Covariance matrix}.--- Now we prove the relation between their covariance matrix with the QFI matrix.

The 3-design property of random Clifford unitaries give us that, for an opeartor $A$ and traceless operators $B$ and $C$ on a $n$-qubit system,
\begin{align}
    \Exp{U}{\bra{\psi}UAU^\dagger\ket{\psi}\bra{\psi}UBU^\dagger\ket{\psi}\bra{\psi}UCU^\dagger\ket{\psi}}
    =\frac{\trace{}{A}\trace{}{BC}+\trace{}{ABC}+\trace{}{ACB}}{2^n(2^n+1)(2^n+2)}
\end{align}
So we can directly calculate the covariance matrix
\begin{align}\label{eq:proof:covariance}
V_{ij} &=
    \Exp{\alpha,x,U,z}{\left(\hat{\theta}_i((x|\alpha),(U,z))-\theta_i^0\right)
    \left(\hat{\theta}_j((x|\alpha),(U,z))-\theta_j^0\right)}\Big|_{\theta=\theta^0}
    \notag\\
    &=\Exp{\alpha,x}{
    \Exp{U,z}{
    \trace{}{X_i^{(x|\alpha)} \hat\rho(U,z)^{(x|\alpha)}}\trace{}{X_j^{(x|\alpha)} \hat\rho(U,z)^{(x|\alpha)}}
    }
    }\Big|_{\theta=\theta^0}
    \notag\\
    &=\bbE_{\alpha,x}\Bigg[
    \sum_{z\in\set{0,1}^r}\bbE_U\Big[
    \bra{z}U\ketbrat{\psi_{x|\alpha}(\theta)}U^\dagger\ket{z}
    \trace{}{(2^r+1)X_i^{(x|\alpha)}U^\dagger\ketbrat{z}U}
    \notag\\
    &\qquad\qquad\cdot
    \trace{}{(2^r+1)X_j^{(x|\alpha)}U^\dagger\ketbrat{z}U}
    \Big]
    \Bigg]\Big|_{\theta=\theta^0}
    \notag\\
    &=\frac{2^r+1}{2^r+2}
    \Exp{\alpha,x}{
    \trace{}{X_i^{(x|\alpha)}X_j^{(x|\alpha)}}+\bra{\psi_{x|\alpha}(\theta)}X_i^{(x|\alpha)}
    X_j^{(x|\alpha)}\ket{\psi_{x|\alpha}(\theta)}+
    \bra{\psi_{x|\alpha}(\theta)}X_j^{(x|\alpha)}
    X_i^{(x|\alpha)}\ket{\psi_{x|\alpha}(\theta)}
    }\Big|_{\theta=\theta^0}
    \notag\\
    &=4\cdot \frac{2^r+1}{2^r+2}\cdot (J^{-1}_M)_{ij}
\end{align}
On the third equality, we use $\trace{}{X^{x|\alpha}_i}=0$.
On the last equality, we use 
\begin{gather}
    \Exp{\alpha,x}{\trace{}{X_i^{(x|\alpha)}X_j^{(x|\alpha)}}}\Big|_{\theta=\theta^0}=
    \Exp{\alpha,x}{
    \bra{\psi_{x|\alpha}(\theta)}X_i^{(x|\alpha)}
    X_j^{(x|\alpha)}\ket{\psi_{x|\alpha}(\theta)}}\Big|_{\theta=\theta^0}\notag\\
    =\sum_{kl}(J_M^{-1})_{ik}(J_M^{-1})_{jl}
    \Exp{\alpha,x}{\trace{}{L^{x|\alpha}_kL^{(x|\alpha)}_l}}\Big|_{\theta=\theta^0}
    =2(J^{-1}_M)_{ij}
\end{gather}

Finally, we bound the relation between $J_M$ and $J$. The central fact to notice is that QFI matrix is the metric for fidelity, i.e, 
\begin{equation}\label{eq:metric}
\abs{\bbraket{\psi(\theta)|\psi(\theta+\mathrm{d}\theta)}}^2
   =
   1-\frac{1}{4}\sum_{ij}J_{ij}(\ket{\psi(\theta)})\,\mathrm{d}\theta_i
    \mathrm{d\theta}_j + o(\norm{\rmd\theta}^2). 
\end{equation}
In the certification protocol, the failure probability is 
\begin{align}
    \Pr \left(\failed\right) & = \sum_{x|\alpha}\Pr \left(\failed|(x|\alpha)\right)
    \Pr\left((x|\alpha)\right)\notag\\
    &=\sum_{x|\alpha}p_\alpha \trace{}{K^{(\alpha)}_x\ketbrat{\psi(\theta)} K^{(\alpha)\dagger}_x}
    \left(1-\left|\braket{\psi_{x|\alpha}(\theta^0)}{\psi_{x|\alpha}(\theta)}\right|^2
    \right)
    \notag\\
    &=1-\sum_{x|\alpha}p_\alpha \trace{}{K^{(\alpha)}_x\ketbrat{\psi(\theta)}K^{(\alpha)\dagger}_x}
    \left|\braket{\psi_{x|\alpha}(\theta^0)}{\psi_{x|\alpha}(\theta)}\right|^2.
\end{align}
Using the soundness condition in Definition.~\ref{def:cfe-certification}, we have
\begin{align}
    1-\sum_{x|\alpha}p_\alpha \trace{}{K^{(\alpha)}_x\ketbrat{\psi(\theta)} K^{(\alpha)\dagger}_x}
    \left|\braket{\psi_{x|\alpha}(\theta^0)}{\psi_{x|\alpha}(\theta)}\right|^2
    \geq 
    \frac{1-\left|\braket{\psi(\theta^0)}{\psi(\theta)}\right|^2}{\Delta}. 
\end{align}
Let 
\begin{equation}
    f(\theta^0)=\left(1-\sum_{x|\alpha}p_\alpha \trace{}{K^{(\alpha)}_x\ketbrat{\psi(\theta)} K^{(\alpha)\dagger}_x}
    \left|\braket{\psi_{x|\alpha}(\theta^0)}{\psi_{x|\alpha}(\theta)}\right|^2\right)
    -\left(
    \frac{1-\left|\braket{\psi(\theta^0)}{\psi(\theta)}\right|^2}{\Delta}\right).
\end{equation}
Then $f(\theta^0)\geq 0$ in the neighborhood of $\theta^0=\theta$, and $f(\theta)=0$. This means that the Hessian $\nabla^2f(\theta^0)$ at $\theta^0=\theta$ is positive semi-definite. Using  
Eq.~\eqref{eq:metric}, we have
\begin{equation}
    J_M\succeq \frac{1}{\Delta}J.
\end{equation}
Together with Eq.\eqref{eq:proof:covariance}, we have
\begin{equation}
    V\preceq 4\cdot\frac{2^r+1}{2^r+2}\cdot \Delta J^{-1}\preceq 4\Delta J^{-1}.
\end{equation}
\end{proof}

In the following, we will use Lemma~\ref{lem:main} with recent developments for certification~\cite{huang2025certifying,gupta2025few,coladangelo2026power} to build our metrology protocols.

\section{Non-adaptive protocol}\label{app:protocol1}

\subsection{Results and remarks}
We first introduce the simplest protocol, where the measurement is performed on the computational basis for all but one (or $\calO(1)$, which will be clear later) randomly chosen qubit. The performance of this protocol is bounded by the mixing time of a Markov chain defined with the base state, which we define as below.
\begin{definition}
    [Markov chain on the Boolean hyper cube]
    \label{def:markov-chain}
    Let $\ket{\psi}$ be a $n$-qubit quantum state, and $p(x)=\abs{\bra{z}\psi\rangle}^2$ be the probability on the computational basis $z\in\set{0,1}^n$. Define the level-$r$ Boolean hypercube as a graph $G=(V,E)$ with vertices labeled by $z$. An edge exists between vertices $z_1$ and $z_2$ when the differ in $1\leq m\leq r$ bits. Then the associated level-$r$ Markov chain of $\ket{\psi}$ is a Markov chain on $G$ with transition matrix
    \begin{equation}
        P(z_1,z_2)=\begin{cases}
\frac{1}{n}\frac{p(z_2)}{p(z_1)+p(z_2)} & (z_1,z_2)\in E,\quad p(z_{1,2})\neq 0 \\
\frac{1}{n}\sum_{z':(z_1,z')\in E}\frac{p(z_1)}{p(z_1)+p(z')}
 & z_1=z_2,\quad p(z_1)\neq 0 \\
 0 & \mathrm{otherwise.}
\end{cases}
    \end{equation}
    Eigenvalues of transition matrix $P$ are at most 1, with at least one eigenvalue exactly being 1. We mixing time as $\tau = 1/(1-\lambda_2)$, where $\lambda_2$ is the second largest eigenvalue.
\end{definition}

\begin{theorem}
    [Quantum metrology with non-adaptive randomized measurement protocol]
    \label{thm:metrology1}
    Let $\calM$ be the single-qubit randomized measurement protocol that measurements the computational basis for all but one randomly chosen qubit, i.e., 
    \begin{gather}
        \calM=\set{\calM_{i,\alpha},\frac{1}{3n}}_{i\in[n],\alpha=x,y,z},\notag\\
        \calM_{i,\alpha}=\Big\{
        \ketbrat{z_1}\otimes\cdots\otimes\ketbrat{z_{i-1}}\otimes \ketbrat{z^{\alpha}_i}\otimes
        \ketbrat{z_{i+1}}\otimes\cdots\otimes \ketbrat{z_n}
        \Big\}_{z\in\set{0,1}^{n}},
    \end{gather}
    where $\ketbrat{z_i}$ is the projector of the $i$-th qubit to the computational basis. $\ketbrat{z^\alpha_i}$ is the projector of the $i$-th qubit to $x,y,z$ basis. The CFI matrix for any target state $\ket{\psi(\theta)}$ satisfies
    \begin{equation}
        I^{-1}(\calM,\ket{\psi(\theta)})\preceq 4\tau(\theta)
        J^{-1}(\ket{\psi(\theta)}).
    \end{equation}
    Here $\tau(\theta)$ is the mixing time of the associated level-1 Markov chain for $\ket{\psi(\theta)}$.
\end{theorem}

We make two important remarks.
\begin{itemize}
    \item There are cases where the level-1 Markov chain never fully mixes. On can define the level-$r$ protocol similar to Theorem~\ref{thm:metrology1} by randomly choosing $r$ unmeasured qubits. The CFI is bounded by the mixing time of level-$r$ Markov chain.
    \item Alternative, one can change the measurement basis to, e.g., $x$ basis to deal with some extreme cases, for example, the Greenberger-Horne-Zeilinger(GHZ) state.
\end{itemize}
Taking into account those variant protocols,
the mixing time for various states is systematically studied in~\cite{huang2025certifying}. We summarize in below. 
\begin{itemize}
    \item With $1-2^{-\calO(n)}$ probability, Haar random states have mixing time $\calO(n^2)$ for $r=1$ (Theorem~16 in Ref.~\cite{huang2025certifying}). Note that we can do better for Haar random states, see Appendix.~\ref{app:haar}.
    \item For phase state $\ket{\psi}=\frac{1}{\sqrt{2^n}}\sum_{z}e^{i\phi_z}\ket{z}$, the mixing time is $n$ for $r=1$ (Lemma~27 in Ref.~\cite{huang2025certifying}).
    \item For $\kappa$-local sign-free (stoquastic) Hamiltonian $H$ with gap $\Delta$ and ground state energy $E_0$, the mixing time of ground state is bounded by (Theorem~30 in Ref.~\cite{huang2025certifying}. See also Ref.~\cite{bravyi2022simulate})
    \begin{equation}
        \tau\leq \frac{1}{\Delta}\left(\sum_{i=1}^\kappa\binom{n}{\kappa}\right)
        \left(\max_{z\in\set{0,1}^n}\bra{z}H\ket{z}-E_0
        \right)
        \lesssim \calO(n^{\kappa+1})
    \end{equation}
    for $r=\kappa$.
    \item For GHZ state, using the $x$ basis protocol, the mixing time is $n/2$ for $r=2$ (Theorem~31 in Ref.~\cite{huang2025certifying}).
\end{itemize}

\subsection{Performance guarantee}

To analyze the performance, we need to prove the performance of the CFE certification protcol using the reduced version of $\calM$ that discarding the single-qubit random rotation.
\begin{lemma}
\label{lem:L}
    Let 
        \begin{equation}
        \calM=\set{\calM_i,\frac{1}{n}}_{i\in[n]},\quad
        \calM_i=\Big\{
        \ketbrat{z_1}\otimes\cdots\otimes\ketbrat{z_{i-1}}\otimes \id_i\otimes
        \ketbrat{z_{i+1}}\otimes\cdots\otimes \ketbrat{z_n}
        \Big\}_{z\in\set{0,1}^{n-1}},
    \end{equation}
    Then the CFE certification protocol using $\calM$ with target state $\ket{\psi}$ has failure probability $1-\trace{}{\rho L}$, $L$ is a $2^n\times 2^n$ matrix with entries
    \begin{equation}\label{eq:L}
       \bra{z_1} L\ket{z_2} = \begin{cases}
            \frac{1}{n}\frac{\sqrt{p(z_1)p(z_2)}}{p(z_1)+p(z_2)}e^{i(\phi_{z_1}-\phi_{z_2})} & (z_1,z_2)\in E,\quad p(z_{1,2})\neq 0  \\
            \frac{1}{n}\sum_{z':(z_1,z')\in E}\frac{p(z_1)}{p(z_1)+p(z')}
            & z_1= z_2,\quad p(z_1)\neq 0\\
            0 & \mathrm{otherwise.}
        \end{cases}
    \end{equation}
    Here $E$ is the set of edges of level-1 Boolean hypercube defined in Def.~\ref{def:markov-chain}, $p(z)$ and $\phi_z$ are defined by
    \begin{equation}
        \ket{\psi}=\sum_{z\in\set{0,1}^n}\sqrt{p(z)}e^{i\phi_z}\ket{z}.
    \end{equation}
\end{lemma}
\begin{proof}
    The failure probability for certification is
    \begin{align}
        \Pr \left(\mathrm{fail}\right) &= \sum_{i\in[n]}\sum_{z\in\set{0,1}^{n-1}} \Pr\left(\mathrm{fail}~|~i,z\right)\Pr\left(\text{sample $i$-th qubit \& measurement outcome $z$}\right)
        \notag\\
        &=\sum_{i\in[n]}\sum_{z\in\set{0,1}^{n-1}}
        \left(1-\bra{\psi_{z|i}}\frac{\bra{z^{(i)}}\rho\ket{z^{(i)}}}{\trace{}{\bra{z^{(i)}}\rho\ket{z^{(i)}}}}
        \ket{\psi_{z|i}}\right)\cdot
        \frac{1}{n}\trace{}{\bra{z^{(i)}}\rho\ket{z^{(i)}}}
        \notag\\
        &=1-\sum_{i\in[n]}\sum_{z\in\set{0,1}^{n-1}}\frac{1}{n}\trace{}{
        \rho\cdot \ketbrat{z^{(i)}}\otimes \ketbrat{\psi_{z|i}}
        }
    \end{align}
    where on the second equality we use $\ket{z^{(i)}}$ to denote the computational basis state specified by $z$ on all but $i$-th qubit, and $\ket{\psi_{z|i}}= \bra{z^{(i)}}\psi\rangle/\norm{\bra{z^{(i)}}\psi\rangle}_2$.
    Here we extract a matrix
    \begin{align}\label{eq:def-L}
        L& = \frac{1}{n}\sum_{i\in[n]}\sum_{z\in\set{0,1}^{n-1}}
        \ketbrat{z^{(i)}}\otimes \ketbrat{\psi_{z|i}},
    \end{align}
    satisfying 
    \begin{align}
        \bra{z}L\ket{z}=\frac{1}{n}\sum_{i\in[n]}
        \frac{p(z)}{p(z)+p(\overline{z}_i)},
    \end{align}
    where $\overline{z}_i$ is the bitstring obtained by flipping the $i$-th bit of $z$. Moreover, when $z_1$ and $z_2$ differ only by $i$-th bit,
    \begin{align}
        \bra{z_1}L\ket{z_2}=\frac{1}{n}\bra{(z_1)_i}\cdot\ketbrat{\psi_{z|i}}
        \cdot \ket{(z_2)_i}=\frac{1}{n}
        \frac{\sqrt{p(z_1)}e^{i\phi_{z_1}}\cdot
        \sqrt{p(z_2)}e^{-i\phi_{z_2}}
        }{p(z_1)+p(z_2)}.
    \end{align}
    When $z_1$ and $z_2$ differs on more than one bit, clearly $\bra{z_1}L\ket{z_2}=0$. In this way, we verified all the properties enforced by Eq.~\eqref{eq:L}
\end{proof}

\begin{lemma}
    \label{lem:property-L}
    The matrix $L$ in Eq.~\eqref{eq:L} and transition matrix $P$ in Def.~\ref{def:markov-chain} have the same spectra.
\end{lemma}
\begin{proof}
    Define the $2^n\times 2^n$ matrices
    \begin{align}
        S = \left(\sum_{z:p(z)\neq 0}\sqrt{p(z)}\ketbrat{z}\right)
        \otimes \id_{\set{z:p(z)=0}},\quad
        F = \left(\sum_{z:p(z)\neq 0}e^{i\phi_z}\ketbrat{z}\right)
        \otimes \id_{\set{z:p(z)=0}}.
    \end{align}
    Then we have
    \begin{align}
        L = F\cdot S\cdot P\cdot S^{-1}\cdot F^{-1}.
    \end{align}
    That is, $L$ is similar to $P$. So they have the same spectra.
\end{proof}

Above lemmas help us to prove Theorem~\ref{thm:metrology1}.
\begin{proof}
    [Proof of Theorem~\ref{thm:metrology1}]
    Since $P$ is a transition matrix, all the eigenvalues $\lambda$ satisfy $|\lambda|\leq 1$, and there is at least on eigenvalue $1$ corresponding to steady state. By Lemma~\ref{lem:property-L}, $L$ has the same properties. Moreover, $L$ is a Hermitian positive semi-definite matrix. So all the eigenvalues are in $[0,1]$. The mixing time for the Markov chain associated with $\ket{\psi}$ is $\tau=1/(1-\lambda_2)$, where $\lambda_2$ is the second largest eigenvalue of $L$.

    By construction (Eq.~\eqref{eq:def-L}), $\ket{\psi}$ is the eigenstate of $L$ with eigenvalue 1. Denote $\ket{\lambda_i}$ for the eigenstate of $L$ with eigenvalue $\lambda_i$ for $i\geq 2$. Using Lemma~\ref{lem:L}, the failure probability of the CFE certification protocol using $\calM$ is
    \begin{align}
        1-\trace{}{\rho L} & = 1-\left(\bra{\psi}\rho\ket{\psi}+\sum_{i\geq 2}\lambda_i \bra{\lambda_i}\rho\ket{\lambda_i}\right)
        \notag\\
        &\geq 1-\left(\bra{\psi}\rho\ket{\psi}+\sum_{i\geq 2}\lambda_2 \bra{\lambda_i}\rho\ket{\lambda_i}\right)\notag\\
        &=1-\left(\bra{\psi}\rho\ket{\psi}+\lambda_2 \trace{}{\rho (\id-\ketbrat{\psi})}\right)\notag\\
        &=(1-\lambda_2)(1-\bra{\psi}\rho\ket{\psi})\notag\\
        &=\frac{1-\bra{\psi}\rho\ket{\psi}}{\tau}
    \end{align}
    Finally, note that for single qubit the random Pauli rotation is the same as random Clifford. So
    by Lemma~\ref{lem:main}, we prove our result.
\end{proof}

\section{Adaptive decision-tree protocol}\label{app:protocol2}

\subsection{Results and remarks}
While the computational basis randomized measurements is conceptually simple to implement, its performance depends on target states. By leveraging the adaptive measurements, we can construct a metrology protocol with $\Delta=n$ for \emph{all} pure states while still using single-qubit measurements.

The core of this protocol is the decision-tree (DT) basis
\begin{definition}
    [Decision-tree basis]
    \label{def:dt-basis}
    Let $\mathsf{T}$ be a depth-$n$ binary tree. Each internal node $v$ is associated with two outgoing edges labeled by orthogonal single-qubit states
$\ket{\psi_v^{(0)}}, \; \ket{\psi_v^{(1)}} \in \mathbb{C}^2$.
Each leaf $\ell$ of $\mathsf{T}$ is uniquely identified by a bitstring
$x = (x_1, \dots, x_n) \in \{0,1\}^n$,
which specifies the path from the root to $\ell$, where $x_i$ indicates the choice of edge at depth $i$.
For a prefix $x_{<i} := (x_1, \dots, x_{i-1})$, let $v_{x_{<i}}$ denote the node reached after following the path specified by $x_{<i}$ from the root.
Then the decision-tree basis $\calT=\set{\ell_x}$ is a set of orthonormal basis associated with each leaf,
\begin{equation}
    \ket{\ell_x}
= \bigotimes_{i=1}^n \ket{\psi_{v_{x_{<i}}}^{(x_i)}}.
\end{equation}
\end{definition}

Note that given an $n$-qubit DT basis, one can measure this basis using adaptive single-qubit measurements. The most important property of DT basis, proved in~\cite{zhou2020saturating,gupta2025few}, is that, for every two pure states, one can always find a DT basis such that they are both phase states (see Lemma~\ref{lem:phase-dt} in the next subsection). 
Recall from the last section that a phase state has mixing time $\tau=n$. This suggests that randomized measurements designed using DT-basis may achieve $\Delta = n$ for all quantum states, although the explicit construction requires additional care, as shown below.

\begin{theorem}
    [Quantum metrology with adaptive decision-tree protocol]
    \label{thm:metrology2}
    Let $x\in\set{0,1}^{k}$ for $k\in[n-1]$ and $x^{0,1}\in\set{0,1}^{k+1}$ be the bitstrings with $k+1$-th bits either $0$ or $1$. Choose the single-qubit measurements $\calM$ as
    \begin{gather}
        \calM=\set{\calM_{i,\alpha},\frac{1}{3n}}_{i\in[n],\alpha=x,y,z},\notag\\
        \calM_{i,\alpha}=\Big\{
        \ketbrat{x_1}\otimes\cdots\otimes\ketbrat{x_{i-1}}\otimes \ketbrat{x^\alpha_i}\otimes
        \ketbrat{\ell^x_{z_1}}\otimes\cdots\otimes \ketbrat{\ell^x_{z_{n-i}}}
        \Big\}_{x\in\set{0,1}^{i},z\in\set{0,1}^{n-i}},
    \end{gather}
    where $\ket{\ell^x_z}\in\calT_x$ is the decision tree basis such that both $\bra{x_{}^0}\psi(\theta)\rangle/\norm{\bra{x_{}^0}\psi(\theta)\rangle}_2$ and $\bra{x_{}^1}\psi(\theta)\rangle/\norm{\bra{x_{}^1}\psi(\theta)\rangle}_2$ are phase states.
    $\ketbrat{x_i}$ is the computational basis projector of the $i$-th qubit.
    $\ketbrat{x^\alpha_i}$ is the projector of $i$-th qubit to $x,y,z$ basis.
    The CFI matrix for any target state $\ket{\psi(\theta)}$ satisfies
    \begin{equation}
        I^{-1}(\calM,\ket{\psi(\theta)})\preceq 4n
        J^{-1}(\ket{\psi(\theta)}).
    \end{equation}
\end{theorem}

\subsection{Performance guarantee}
We first justify that the decision-tree basis needed in Theorem~\ref{thm:metrology2} exists.
\begin{lemma}
    [Phase state in DT basis, cf.~Corollary.~7 in~\cite{gupta2025few} or Lemma~1 in~\cite{zhou2020saturating}]
    \label{lem:phase-dt}
    For every $n$-qubit states $\rho_1$ and $\rho_2$, there is an associated DT basis $\calT=\set{\ket{\ell_x}}$, such that the diagonal matrix entries are all $1/2^n$ for both states, i.e.,
    \begin{equation}
        \bra{\ell_x}\rho_1\ket{\ell_x}=\bra{\ell_x}\rho_2\ket{\ell_x}=\frac{1}{2^n},\quad
        \forall \ket{\ell_x}\in\calT
    \end{equation}
\end{lemma}
\begin{proof}
    For every single-qubit mixed states $\sigma_1$ and $\sigma_2$, there is a corresponding basis $\set{\ket{b},\ket{b^\perp}}$, such that the measurement probabilities of these two bases on both states are equally $1/2$. This is because every single-qubit state can be represented as a vector on the Bloch sphere. It suffices to choose $\ket{b}$ to the orthogonal directions of vectors of  $\rho_1$ and $\rho_2$. Then for $\rho_1$, $\rho_2$, and $i=1,\cdots,n$-th qubits, we can inductively build a DT such that the adaptive measurements of each node to its descendants yields equal probabilities $1/2$. So the probability of measuring each leaf under this DT basis is $1/2^n$ for both $\rho_1$ and $\rho_2$. 
\end{proof}

Note that when $\rho_1=\ketbrat{\psi}$ is a pure state, it is a phase state under the DT basis,
\begin{equation}
    \ket{\psi}=\frac{1}{\sqrt{2^n}}\sum_{x\in\set{0,1}^n}e^{i\phi_x}\ket{\ell_x}.
\end{equation}

We also point out that Lemma~\ref{lem:phase-dt} cannot be generalized to qudit cases with $d\geq 3$ where we require $d$ distinct states to be phase states. Consider the following counterexample. Let $\{\ket{0},\ket{1},\ket{2}\}$ be a set of orthonormal basis. Consider
\begin{align}
    \ket{\psi_1}=\ket{0},\quad \ket{\psi_2}=\frac{\ket{0}+\ket{1}}{\sqrt{2}},
    \quad \ket{\psi_3}=\ket{2}.
\end{align}
Assume there exists a set of basis $\{\ket{e_0},\ket{e_1},\ket{e_2}\}$ such that $\ket{\psi_{1,2,3}}$ are all phase states. From $\ket{\psi_{1,3}}$, the basis set should satisfy $|\bra{e_k}0\rangle|=|\bra{e_k}2\rangle|=1/\sqrt{3}$ for all $k$. Due to normalization, $|\bra{e_k}1\rangle|=1/\sqrt{3}$ for all $k$. This means that $\bra{e_k}j\rangle=e^{i\theta_{kj}}/\sqrt{3}$ for all $k$ and $j$. 
Now we consider $\ket{\psi_2}$. By assumption,
\begin{align}
    |\bra{e_k}\psi_2\rangle|^2=\frac{1}{6}|e^{i\theta_{k0}}+e^{i\theta_{k1}}|^2=\frac{1}{3},\quad \forall k.
\end{align}
As a result,
\begin{align}
    \theta_{k0}-\theta_{k1}=\pm\frac{\pi}{2},\quad \forall k.
\end{align}
Due to orthonormality,
\begin{align}
    \sum_{j=0}^2 e^{i(\theta_{k_1j}-\theta_{k_2j})}=0,\quad\forall k_1\neq k_2.
\end{align}
Note that
\begin{align}
    (\theta_{k_10}-\theta_{k_20})-(\theta_{k_11}-\theta_{k_21})\in\{0,\pi\}.
\end{align}
So $e^{i(\theta_{k_10}-\theta_{k_20})}=\pm e^{i(\theta_{k_11}-\theta_{k_21})}$. As a result, $|e^{i(\theta_{k_12}-\theta_{k_22})}|=0$ or $2$. Contradiction.

Now we rigorously prove Theorem~\ref{thm:metrology2}. To this end, we analyze the performance of corresponding certification protocol.
\begin{lemma}
\label{lem:ryan-protocol-performance}
    Let 
    \begin{gather}
        \calM=\set{\calM_i,\frac{1}{n}}_{i\in[n]},\notag\\
        \calM_i=\Big\{
        \ketbrat{x_1}\otimes\cdots\otimes\ketbrat{x_{i-1}}\otimes \id_i\otimes
        \ketbrat{\ell^x_{z_1}}\otimes\cdots\otimes \ketbrat{\ell^x_{z_{n-i}}}
        \Big\}_{x\in\set{0,1}^{i-1},z\in\set{0,1}^{n-i}}
    \end{gather}
    be the measurements defined by Theorem~\ref{thm:metrology2}. Then the CFE certification protocol by using $\calM$ has a certification gap $\Delta=n$.
\end{lemma}
\begin{proof}
    The proof follows from Lemma 7 in~\cite{gupta2025few}.
    Let $\ket{\psi}$ and $\rho$ denote the target state and lab state, respectively. Later, we will focus on the case where $\rho=\ketbrat{\phi}$ is a pure state, since this is sufficient for our metrological purpose.
    Define the potential difference function 
    \begin{equation}
       \Gamma_k = \bbE_{x \in \{0,1\}^k }\bra{\psi_x}\rho_x\ket{\psi_x} - \bbE_{x \in \{0,1\}^{k-1} }\bra{\psi_x}\rho_x\ket{\psi_x}, 
    \end{equation}
    where $\ket{\psi_x}$ and $\rho_x$ are normalized post-selected states of the last $n - \abs{x}$ qubits conditioned on the first $\abs{x}$ measurement outcomes  being $x$. When $\abs{x} = n$, $\ket{\psi_x} = \rho_x = 1$.  Clearly, 
    \begin{equation}
        \label{eq:potential}
    \sum_{k=1}^n \Gamma_k = 1 - \bra{\psi}\rho\ket{\psi}. 
    \end{equation}
    
    Next, we show for any $x \in \{0,1\}^{k-1}$, 
    \begin{equation}
    \label{eq:lemma7}
    \bbE_{m \in \{0,1\}}\bra{\psi_{x^m}}\rho_{x^m}\ket{\psi_{x^m}} - \bra{\psi_{x}}\rho_{x}\ket{\psi_{x}} \leq \sum_{z} p_{z|(k,x)} \left( 1 - \bra{\psi_{xz}}\rho_{xz}\ket{\psi_{xz}}) \right), 
    \end{equation}
    On the left hand side of the inequality, recall that $x^m$ represents the bitstring with first $k-1$ bits being $x$, and the $k$-th bit is $m$. On the right hand side,
    $z \in \{0,1\}^{n-k}$ represents the last $n-k$ digits of measurement outcomes using the DT basis $\ket{\ell^x_z}$, $\ket{\psi_{xz}}$ ($\rho_{xz}$) denotes the normalized post-selected states of the  $k$-th qubit conditioned on the first $k-1$ measurement outcomes  being $x$, and last $n-k$ measurement outcomes being $z$. 
    $p_{z|(k,x)} =\trace{}{ \bra{\bigotimes\ell^x_{z_i}}\rho_{x}\ket{\bigotimes\ell^x_{z_i}}}$ represents the probability getting measurement outcomes $z$ on the last $n-k$ qubits.  From \eqref{eq:lemma7},
    \begin{equation}
        \begin{split}
    \Gamma_k 
    &= \bbE_{x\in \{0,1\}^{k-1}}\left[\bbE_{m \in \{0,1\}}\bra{\psi_{x^m}}\rho_{x^m}\ket{\psi_{x^m}} - \bra{\psi_{x}}\rho_{x}\ket{\psi_{x}}\right] \\
    &\leq \sum_{x \in\{0,1\}^{k-1}} p_{x|k} \sum_{z \in\{0,1\}^{n-k}} p_{z|(k,x)} \left( 1 - \bra{\psi_{xz}}\rho_{xz}\ket{\psi_{xz}}) \right)\\
    &= \sum_{x \in \{0,1\}^{k-1},z\in\set{0,1}^{n-k}} \frac{p_{k,xz}}{p_k} \left( 1 - \bra{\psi_{xz}}\rho_{xz}\ket{\psi_{xz}}) \right) \\
    &= n \sum_{x \in \{0,1\}^{k-1},z\in\set{0,1}^{n-k}} p_{k,xz} \left( 1 - \bra{\psi_{xz}}\rho_{xz}\ket{\psi_{xz}}) \right), 
    \end{split}
    \end{equation}
    where $p_k = 1/n$ because $k$ is selected from uniformly randomly. 
    From \eqref{eq:potential}, 
    \begin{align}
        1 - \bra{\psi}\rho\ket{\psi} = \sum_{k=1}^n \Gamma_k \leq n \sum_{k,x} p_{k,xz} \left( 1 - \bra{\psi_{xz}}\rho_{xz}\ket{\psi_{xz}}) \right),
    \end{align}
    proving the lemma.
    
    To prove \eqref{eq:lemma7}, we directly compute both sides of the inequality. Assume $\rho=\ketbrat{\phi}$ is a pure state, and 
    let $x \in \{0,1\}^k$.  Denote
    \begin{equation}
        \ket{\psi_x} = \ket{0} \otimes \ket{u^0} + \ket{1} \otimes \ket{u^1},\quad 
    \ket{\phi_x} = \ket{0} \otimes \ket{v^0} + \ket{1} \otimes \ket{v^1},
    \end{equation}
    where $\ket{0},\ket{1}$ forms the orthonormal basis of the $k$-th qubit. $\ket{u^{0,1}}$ and $\ket{v^{0,1}}$ are unnormalized states.
    We also introduce the following notations,
    \begin{gather}
        u^{0,1}_z = \braket{u^{0,1}}{\ell^x_z},\quad 
    v^{0,1}_z = \braket{v^{0,1}}{\ell^x_z},\quad \zeta_z = \frac{u^{1}_z}{u^{0}_z}.
    \end{gather}
    When one of $\ket{u^{0,1}}$, $\ket{v^{0,1}}$ is 0, left side of \eqref{eq:lemma7} is 0. The inequality trivially holds. We assume that none of $\ket{u^{0,1}}$, $\ket{v^{0,1}}$ is 0.
    
    Since by requirements, $\ket{u^{0,1}}$ are phase states,
    \begin{gather}
    \abs{\zeta_z} = \abs{\zeta_{z'}} =: \zeta,
    \quad 
    \abs{u^{0}_z}^2 = \frac{1}{2^{n-k}(1+\zeta^2)},\quad 
    \abs{u^{1}_z}^2 = \frac{\zeta^2}{2^{n-k}(1+\zeta^2)}.
    \end{gather}
    The right hand side of \eqref{eq:lemma7} reads
    \begin{align}
       \bra{\psi_{xz}}\rho_{xz}\ket{\psi_{xz}}  = \frac{\abs{\braket{\psi_x }{ \ell^x_z}\braket{\ell^x_z}{\phi_x}}^2}{\norm{\braket{\psi^x }{ \ell^x_z}}^2\norm{\braket{\ell^x_z}{\phi_x}}^2} = \frac{1}{1+\zeta^2} \frac{\abs{v_z^0 + \zeta_z^* v_z^1}^2}{p_{z|(k,x)}}
    \end{align}
    \begin{align}
        \sum_{z} p_{z|(k,x)} \left( 1 - \bra{\psi_{xz}}\rho_{xz}\ket{\psi_{xz}}) \right) = 1 - \frac{1}{1+\zeta^2} \sum_z \abs{v_z^0 + \zeta_z^* v_z^1}^2 = \frac{1}{1 + \zeta^2}\sum_z \abs{\zeta_z v_z^0 - v_z^1}^2. 
    \end{align}
    The left hand side of \eqref{eq:lemma7} reads 
        \begin{align}
    \bbE_{m \in \{0,1\}}\bra{\psi_{x^m}}\rho_{x^m}\ket{\psi_{x^m}} - \bra{\psi_{x}}\rho_{x}\ket{\psi_{x}}
    &= \abs{\braket{\psi_{x^0}}{v^0}}^2 + \abs{\braket{\psi_{x^1}}{v^1}}^2 - \abs{\braket{u^0}{v^0} + \braket{u^1}{v^1}}^2\notag\\
    &= (1+\zeta^2)\abs{\braket{u^0}{v^0}}^2+\left(1+\frac{1}{\zeta^2}\right)\abs{\braket{u^1}{v^1}}^2
    - \abs{\braket{u^0}{v^0} + \braket{u^1}{v^1}}^2
    \notag\\
    &= \abs{\zeta \sum_z (u_z^0)^* v_z^0 - \frac{1}{\zeta} \sum_z (u_z^1)^* v_z^1}^2 \notag\\
    &= \abs{\frac{1}{\zeta}\sum_z (u_z^1)^* \zeta_z v_z^0 - \frac{1}{\zeta} \sum_y (u_z^1)^* v_z^1}^2 \notag\\
    &\leq  \frac{1}{\zeta^2}\sum_z \abs{u_z^1}^2 \sum_z \abs{\zeta_z v_z^0 - v_z^1}^2 = \frac{1}{1 + \zeta^2}\sum_z \abs{\zeta_z v_z^0 -  v_z^1}^2.
        \end{align}
    \eqref{eq:lemma7} is then proven. 
\end{proof}

Theorem~\ref{thm:metrology2} is then a direct consequence of Lemma~\ref{lem:main} and Lemma~\ref{lem:ryan-protocol-performance}.

\section{Protocols for Haar random states}\label{app:haar}

When targeting Haar random states, a polynomial overhead (in previous sections) can be reduced to a logarithmic or constant overhead. 

\subsection{The randomized Pauli measurement protocol}\label{app:haar1}
Recent progress~\cite{du2025certifying,coladangelo2026power} prove that with randomized Pauli measurement, a CFE-based protocol can certify typical Haar-random states. We can use the same measurement protocol to perform metrology on Haar-random states. Note that, performing randomized Pauli measurement on $n-1$ qubits with randomized Clifford measurement on the remaining one qubit is equivalent to perform full randomzied Pauli measurement on all qubits. So the description of measurement need not an explicit CFE structure.

\begin{theorem}
    [Quantum metrology for Haar random states with randomized Pauli measurement]
    \label{thm:metrology4}
    Choose the randomized measurement protocol to be
    \begin{gather}
        \calM=\set{\calM_{\alpha_x,\cdots,\alpha_n},\frac{1}{3^n}}_{\alpha_i=x,y,z},\notag\\
        \calM_{\alpha_1,\cdots,\alpha_n}=\Big\{\ketbrat{x^{\alpha_1}_1}\otimes 
        \ketbrat{x^{\alpha_2}_2}\otimes\cdots\otimes \ketbrat{x^{\alpha_n}_n}
        \Big\}_{x\in\set{0,1}^{n}},
    \end{gather}
    where $\ketbrat{x^{\alpha_i}_i}$ is the projector of $i$-th qubits to $\alpha_i=x,y,z$ basis.
    When the target state $\ket{\psi(\theta)}$ is chosen from Haar random states, there exists a constant $C$ such that this measurement protocol achieves CFI
    \begin{equation}
        I^{-1}(\calM,\ket{\psi(\theta)})\preceq CJ^{-1}(\psi(\theta))
    \end{equation}
    with probability at least $1-e^{-\Omega(n)}$.
\end{theorem}

Theorem~\ref{thm:metrology4} can be proved directly from Lemma~\ref{lem:main} and the following lemma (Theorem~5 of \cite{du2025certifying} or Theorem~1 of~\cite{coladangelo2026robust}),
\begin{lemma}
    [Constant certification gap for Haar random states when $r=1$, rephrased Theorem~5 in~\cite{du2025certifying}]
    Consider the following measurement protocol,
    \begin{gather}
        \calM=\set{\calM_{\alpha_x,\cdots,\alpha_{n-1}},\frac{1}{3^{n-1}}}_{\alpha_i=x,y,z},\notag\\
        \calM_{\alpha_1,\cdots,\alpha_n}=\Big\{\ketbrat{x^{\alpha_1}_1}\otimes 
        \ketbrat{x^{\alpha_2}_2}\otimes\cdots\otimes \ketbrat{x^{\alpha_{n-1}}_{n-1}}
        \Big\}_{x\in\set{0,1}^{n-1}},
    \end{gather}
    where $\ketbrat{x^{\alpha_i}_i}$ is the projector of $i$-th qubits to $\alpha_i=x,y,z$ basis.
    The CFE certification protocol using $\calM$ for target Haar random states $\ket{\psi}$ has a constant certification gap with probability at least $1-e^{-\Omega(n)}$.
\end{lemma}

\subsection{The two-bases protocol}\label{app:haar2}
Built from another recent result~\cite{coladangelo2026power},
using slightly simpler single-qubit measurement on $n-r$ qubits with $r\geq \log n$, one can perform metrology on Haar-random states with $\Delta\approx 2\log n$ by using single-qubit measurements, and $\Delta\approx 2$ by using $\mathcal{O}(\log n)$-qubit measurements. 


The protocol works by measuring the first $n-\lceil \gamma\log n\rceil$ qubits in either $x$ or $z$ basis, and measure the classical shadow on the remaining $\lceil \gamma\log n\rceil$ qubits, for any $\gamma>1$.
\begin{theorem}
    [Quantum metrology for Haar random states with two-bases measurement]
    \label{thm:metrology3}
    Let $\mathrm{Cl}(n)$ be the $n$-qubit Clifford unitary group. Let $\gamma>1$ be a constant, and denote $n_0= \lceil \gamma\log n\rceil$.
    Choose the randomized measurement protocol to be
    \begin{gather}
        \calM=\set{\calM_{U,\alpha},\frac{1}{2}p(U)}_{U\in\mathrm{Cl}(n),\alpha=x,y},\notag\\
        \calM_{U,\alpha}=\Big\{
        \ketbrat{x^{\alpha}_1}\otimes\cdots\otimes\ketbrat{x^{\alpha}_{n-n_0}}\otimes 
        \left(U \left(\otimes_{j=n-n_0}^n\ketbrat{z_j}\right)U^\dagger\right)
        \Big\}_{x\in\set{0,1}^{n-n_0},z\in\set{0,1}^{n_0}},
    \end{gather}
    where $\ketbrat{x^\alpha_i}$ is the projector of $i$-th qubits to either $x$ ($\alpha=x$) or $z$ ($\alpha=z$) basis. $p(U)$ is the uniform distribution over $\mathrm{Cl}(n)$. When the target state $\ket{\psi(\theta)}$ is chosen from Haar random states, this measurement protocol achieves CFI
    \begin{equation}
        I^{-1}(\calM,\ket{\psi(\theta)})\preceq (8+o(1))J^{-1}(\psi(\theta))
    \end{equation}
    with probability at least $1-\calO(2^{-n})$ for any $\gamma$.
\end{theorem}

Note that there are two differences from previous protocols. First, randomly chosen unmeasured qubit is not required. We can fix the unmeasured qubits to be the last $\lceil \gamma\log n\rceil$ qubits. Second, instead of measuring a fix local basis, this Theorem~\ref{thm:metrology3} requires randomized measurement on either $x$ or $z$ basis.

Theorem~\ref{thm:metrology3} can be proved directly from Lemma~\ref{lem:main} and the following lemma (Theorem~5 of \cite{coladangelo2026power}),
\begin{lemma}
    [Soundness of $\calO(\log n)$-qubit certification protocol, rephrased Theorem~5 in~\cite{coladangelo2026power}]
    Consider the following measurement protocol,
    \begin{gather}
        \calM=\set{\calM_{\alpha},\frac{1}{2}}_{\alpha=x,y},\quad
        \calM_{i,\alpha}=\Big\{
        \ketbrat{x^{\alpha}_1}\otimes\cdots\otimes\ketbrat{x^{\alpha}_{n-n_0}}\otimes 
        \id_{n-n_0}
        \Big\}_{x\in\set{0,1}^{n-n_0}}.
    \end{gather}
    The CFE certification protocol using $\calM$ for target Haar-random states $\ket{\psi}$ satisfies the following soundness condition:
    \begin{equation}
        \Pr(\failed)\geq \frac{1-\bra{\psi}\rho\ket{\psi}}{2+o(1)}
    \end{equation}
    for any lab state $\rho$ with probability at least $1-\calO(2^{-n})$.
\end{lemma}

Combining Theorem~\ref{thm:metrology2} and Theorem~\ref{thm:metrology3}, we can establish a metrology protocol for Haar random states that uses only single-qubit measurements but achieves $\Delta = \calO(\log n)$, presenting an exponential improvement over Theorem~\ref{thm:metrology1} and \ref{thm:metrology2}.
\begin{corollary}\label{coro:log-protocol}
    There exists a single-qubit adaptive randomized measurement protocol $\calM$, such that when $\ket{\psi(\theta)}$ is sampled from Haar random states,
    \begin{equation}
        I^{-1}(\calM,\ket{\psi(\theta)})\preceq (8+o(1))\gamma\log n \cdot J^{-1}(\ket{\psi(\theta)})
    \end{equation}
    with probability at least $1-\calO(2^{-n})$.
\end{corollary}
\begin{proof}
    This is achieved by a protocol with two steps:
    \begin{itemize}
        \item First, follow Theorem~\ref{thm:metrology3}, randomly measure the first $n-\lceil \gamma\log n\rceil$ qubits on $x$ or $z$ basis randomly, and obtain a post-selected state $\ket{\psi_{x,z}}$.
        \item Then, applying the POVM constructed in Theorem~\ref{thm:metrology2}. Replace the base state by $\ket{\psi_{x,z}}$.
    \end{itemize}
    By Theorem~\ref{thm:metrology2}, the measurment protocol of the second step yields a CFI
    \begin{equation}
        I^{-1}(\calM_{x,z},\ket{\psi_{x,z}})\preceq 4\gamma \log n \cdot J^{-1}(\ket{\psi_{x,z}}).
    \end{equation}
    By Theorem~\ref{thm:metrology3} and the proof of Lemma~\ref{lem:main}, the first step yields an averaged QFI that satisfies
    \begin{equation}
        4\gamma \log n\cdot \sum_{x,z}p_{x,z} I(\calM_{x,z},\ket{\psi_{x,z}})\succeq J_M=\sum_{x,z}p_{x,z}J(\ket{\psi_{x,z}})\succeq \frac{1}{2+o(1)}J(\ket{\psi(\theta)}).
    \end{equation}
    As a result,
    \begin{align}
        I(\calM,\ket{\psi(\theta)})=I(p_{x,z})+\sum_{x,z}p_{x,z} I(\calM_{x,z},\ket{\psi_{x,z}})
        \succeq \frac{1}{(8+o(1))\gamma \log n}J(\ket{\psi(\theta)}).
    \end{align}
    This proves the result.
\end{proof}

\end{document}